\def\year{2021}\relax
\documentclass[letterpaper]{article} 
\usepackage{aaai21}  
\usepackage{times}  
\usepackage{helvet} 
\usepackage{courier}  
\usepackage[hyphens]{url}  
\usepackage{graphicx} 
\usepackage{graphicx}  
\usepackage{natbib}  
\usepackage{caption} 
\usepackage{booktabs}
\usepackage{algorithm}
\usepackage{algpseudocode}
\usepackage{bm}
\usepackage{amsmath,amssymb,amsthm}
\newtheorem{lemma}{Lemma}
\newtheorem{proposition}{Proposition}
\newtheorem{definition}{Definition}
\newtheorem{theorem}{Theorem}
\newenvironment{sketch}{%
  \proof}{\endproof}
\usepackage{hhline}
\usepackage[font=footnotesize,subrefformat=parens]{subcaption}
\usepackage{arydshln}
\usepackage{tikz}
\usetikzlibrary{calc}
\usetikzlibrary{positioning}
\usetikzlibrary{shapes.geometric}
\usetikzlibrary{patterns}
\newcommand{\vertexsize}{0.7cm}
\newcommand{\smallvertexsize}{0.25cm}
\newcommand{\middlevertexsize}{0.5cm}
\tikzset{
vertex/.style={draw,rectangle,minimum size=\vertexsize,line width=0.1mm},
agent/.style={circle,draw,minimum size=0.5cm,inner sep=0mm,text=black},
agenthighest/.style={agent,fill={rgb:red,0;green,1;blue,3},text=white},
agentstuck/.style={agent,fill={rgb:red,3;green,1;blue,0},text=white},
agentchange/.style={agent,pattern=north east lines,pattern color={rgb:red,0;green,1;blue,3},text=black},
agentgreen/.style={agent,fill={rgb:green,1;blue,1},text=white},
dest/.style={circle,draw,black,densely dashed, minimum size=0.6cm,line width=0.1mm},
smalldest/.style={agent, densely dashed, minimum size=0.4cm},
line/.style={black},
move/.style={->, line, thick},
bt/.style={->, double},
destarrow/.style={->,thick},
state/.style={fill={black},circle,inner sep=0pt,minimum size=3pt},
w-arrow/.style={->,thick,color={rgb:red,3;yellow,1;blue,1}},
vertex-s/.style={vertex,minimum size=\smallvertexsize,line width=0.1mm},
vertex-s-obj/.style={vertex-s,fill=black},
agent-s/.style={agent,minimum size=0.2cm,fill=black},
dest-s/.style={smalldest,minimum size=0.1cm,line width=0.05mm},
destarrow-s/.style={->,line width=0.01cm},
vertex-m/.style={vertex,minimum size=\middlevertexsize,line width=0.1mm},
agent-m/.style={agent,minimum size=0.4cm,fill=black},
destarrow-m/.style={->,thick},
}
\usepackage{macro}
\nocopyright

\title{Time-Independent Planning for Multiple Moving Agents}
\author{
    Keisuke Okumura,
    Yasumasa Tamura,
    Xavier D\'{e}fago \\
}
\affiliations{
  School of Computing, \\Tokyo Institute of Technology,\\
  Tokyo, Japan\\
  \{okumura.k, tamura, defago\}@coord.c.titech.ac.jp
}

\begin{document}

\maketitle
\begin{abstract}
  Typical Multi-agent Path Finding (MAPF) solvers assume that agents move synchronously, thus neglecting the reality gap in timing assumptions, e.g., delays caused by an imperfect execution of asynchronous moves.
  So far, two policies enforce a robust execution of MAPF plans taken as input: either by forcing agents to synchronize or by executing plans while preserving temporal dependencies.
  This paper proposes an alternative approach, called time-independent planning, which is both online and distributed.
  We represent reality as a transition system that changes configurations according to atomic actions of agents, and use it to generate a time-independent schedule.
  Empirical results in a simulated environment with stochastic delays of agents' moves support the validity of our proposal.
\end{abstract}

\section{Introduction}
Multi-agent systems with physically moving agents are becoming gradually more common, e.g., automated warehouse~\cite{wurman2008coordinating}, traffic control~\cite{dresner2008multiagent}, or self-driving cars.
In such systems, agents must move smoothly without colliding.
This is embodied by the problem of Multi-agent Path Finding (MAPF)~\cite{stern2019multi}.
Planning techniques for MAPF have been extensively studied in recent years.

The output of such planning is bound to be executed in real-world situations with agents (robots).
Typical MAPF is defined in discrete time.
Agents are assumed to do two kinds of atomic actions synchronously:
move to a neighboring location or stay at their current location.
Perfect executions for the planning are however difficult to ensure since timing assumptions are inherently uncertain in reality, due to the difficulty of:
1)~accurately predicting the temporal behavior of many aspects of the system, e.g., kinematics,
2)~anticipating external events such as faults and interference, and
3)~ensuring a globally consistent and accurate notion of time in the face of clock shift and clock drift.
Even worse, the potential of unexpected interference increases with the number of agents,
hence the need to prepare for imperfect executions regarding the timing assumptions.

There are two intuitive ways to tackle imperfect executions of MAPF plans taken as input.
The first, and conservative idea is to forcibly synchronize agents' moves, globally or locally.
Most decentralized approaches to MAPF take this approach implicitly~\cite{wiktor2014decentralized,kim2015discof+,okumura2019priority,wang2020walk}.
This policy negatively affects the entire performance of the system with unexpected delays and lacks flexibility~\cite{ma2017multi}.
The second policy makes agents preserve temporal dependencies of the planning~\cite{honig2016multi,ma2017multi,honig2019persistent,atzmon2020robust}.
Two types of temporal dependencies exist:
1)~internal events within one agent and,
2)~order relation of visiting one node.
This policy is sound but still vulnerable to delays.
Consider an extreme example where one agent moves very slowly or crashes.
Due to the second type of dependencies, the locations where the agent will be are constrained by the use of the other agents.
Thus, the asynchrony of the movements is sensitive to the whole system.

We, therefore, propose an alternative approach, called time-independent planning, that aims at online and distributed execution focusing on agents' moves.
We represent the whole system as a transition system that changes configurations according to atomic actions of agents, namely,
1) request the next locations (\requesting),
2) move (\extended), and,
3) release the past locations, or stay (\contracted).
In this time-independent model, any \apriori knowledge for timings of atomic actions is unavailable, representing non-deterministic behaviors of the external environment.
The challenge is to design algorithms tolerant of all possible sequences of actions.

{
  \newcommand{\colwidth}{0.175\hsize}
\begin{figure*}[ht!]
  \centering
  \begin{tabular}{ccccc}
    \begin{minipage}[t]{\colwidth}
      \centering
      \begin{tikzpicture}
        \node[vertex](v1) at (0, 0) {};
        \node[vertex,right=0cm of v1](v2) {};
        \node[vertex,right=0cm of v2](v3) {};
        \node[vertex,right=0cm of v3](v4) {};
        \node[vertex,above=0cm of v1](v5) {};
        \node[vertex,right=0cm of v5](v6) {};
        \node[vertex,right=0cm of v6](v7) {};
        \node[vertex,right=0cm of v7](v8) {};
        \node[agenthighest](a0) at (v6) {$a_7$};
        \node[anchor=east](highlabel) at (0.47, 0.92) {\scriptsize {\tiny priority:}};
        \node[anchor=east](highlabel) at (0.47, 0.7) {\scriptsize \textbf{high}};
        \node[agent](a1) at (v7) {$a_6$};
        \node[agent](a2) at (v3) {$a_5$};
        \node[agent,label=below:{\scriptsize\color{white}\textbf{hg}}](a3) at (v4) {$a_4$};
        \node[agentstuck](a4) at (v8) {$a_3$};
        \node[agent,label=below:{\scriptsize low}](a5) at (v2) {$a_1$};
        \node[agent,label=below:{\scriptsize medium}](a6) at (v1) {$a_2$};
        \node[dest] (dest0) at (a1) {};
        \node[dest] (dest1) at (a2) {};
        \node[dest] (dest2) at (a3) {};
        \node[dest] (dest3) at (a4) {};
        \draw[move] (a0)--(dest0);
        \draw[move] (a1)--(dest1);
        \draw[move] (a2)--(dest2);
        \draw[move] (a3)--(dest3);
      \end{tikzpicture}
      \subcaption{priority inheritance (PI)}
      \label{fig:pibt:init}
    \end{minipage}
    &
    \begin{minipage}[t]{\colwidth}
      \centering
      \begin{tikzpicture}
        \node[vertex](v1) at (0, 0) {};
        \node[vertex,right=0cm of v1](v2) {};
        \node[vertex,right=0cm of v2](v3) {};
        \node[vertex,right=0cm of v3](v4) {};
        \node[vertex,above=0cm of v1](v5) {};
        \node[vertex,right=0cm of v5](v6) {};
        \node[vertex,right=0cm of v6](v7) {};
        \node[vertex,right=0cm of v7](v8) {};
        \node[agenthighest,label=above:{\scriptsize}](a0) at (v6) {$a_7$};
        \node[agent](a1) at (v7) {$a_6$};
        \node[agentchange](a2) at (v3) {$a_5$};
        \node[agent](a3) at (v4) {$a_4$};
        \node[agentstuck](a4) at (v8) {$a_3$};
        \node[agent,label=below:{\scriptsize low (as \textbf{high})}](a5) at (v2) {$a_1$};
        \node[agent](a6) at (v1) {$a_2$};
        \draw[bt] (a4)--(a3);
        \draw[bt] (a3)--(a2);
        \node[dest] (dest2) at (a5) {};
        \node[dest] (dest5) at (a6) {};
        \node[smalldest] (dest6) at (v5) {};
        \draw[move] (a2)--(dest2);
        \draw[move] (a5)--(dest5);
        \draw[move] (a6)--(dest6);
      \end{tikzpicture}
      \subcaption{backtracking and PI again}
      \label{fig:pibt:btpi}
    \end{minipage}
    &
    \begin{minipage}[t]{\colwidth}
      \centering
      \begin{tikzpicture}
        \node[vertex](v1) at (0, 0) {};
        \node[vertex,right=0cm of v1](v2) {};
        \node[vertex,right=0cm of v2](v3) {};
        \node[vertex,right=0cm of v3](v4) {};
        \node[vertex,above=0cm of v1](v5) {};
        \node[vertex,right=0cm of v5](v6) {};
        \node[vertex,right=0cm of v6](v7) {};
        \node[vertex,right=0cm of v7](v8) {};
        \node[agenthighest](a0) at (v6) {$a_7$};
        \node[agent](a1) at (v7) {$a_6$};
        \node[agentchange](a2) at (v3) {$a_5$};
        \node[agent,label=below:{\scriptsize\color{white}\textbf{hg}}](a3) at (v4) {$a_4$};
        \node[agent](a4) at (v8) {$a_3$};
        \node[agent](a5) at (v2) {$a_1$};
        \node[agent](a6) at (v1) {$a_2$};
        \draw[bt] (a6)--(a5);
        \draw[bt] (a5)--(a2);
        \draw[bt] (a2)--(a1);
        \draw[bt] (a1)--(a0);
      \end{tikzpicture}
      \subcaption{backtracking}
      \label{fig:pibt:bt}
    \end{minipage}
    &
    \begin{minipage}[t]{\colwidth}
      \centering
      \begin{tikzpicture}
        \node[vertex](v1) at (0, 0) {};
        \node[vertex,right=0cm of v1](v2) {};
        \node[vertex,right=0cm of v2](v3) {};
        \node[vertex,right=0cm of v3](v4) {};
        \node[vertex,above=0cm of v1](v5) {};
        \node[vertex,right=0cm of v5](v6) {};
        \node[vertex,right=0cm of v6](v7) {};
        \node[vertex,right=0cm of v7](v8) {};
        \node[agenthighest](a0) at (v7) {$a_7$};
        \node[agent](a1) at (v3) {$a_6$};
        \node[agent](a2) at (v2) {$a_5$};
        \node[agent,label=below:{\scriptsize\color{white}\textbf{hg}}](a3) at (v4) {$a_4$};
        \node[agentstuck](a4) at (v8) {$a_3$};
        \node[agent](a5) at (v1) {$a_1$};
        \node[agent](a6) at (v5) {$a_2$};
      \end{tikzpicture}
      \subcaption{one timestep later}
      \label{fig:pibt:result}
    \end{minipage}
    &
    \begin{minipage}[t]{\colwidth}
      \centering
      \begin{tikzpicture}
        \node[vertex](v1) at (0, 0) {};
        \node[vertex,right=0cm of v1](v2) {};
        \node[vertex,right=0cm of v2](v3) {};
        \node[vertex,right=0cm of v3](v4) {};
        \node[vertex,above=0cm of v1](v5) {};
        \node[vertex,right=0cm of v5](v6) {};
        \node[vertex,right=0cm of v6](v7) {};
        \node[vertex,right=0cm of v7](v8) {};
        \node[agenthighest](a0) at (v6) {$a_7$};
        \node[agent](a1) at (v7) {$a_6$};
        \node[agent](a2) at (v3) {$a_5$};
        \node[agent,label=below:{\scriptsize\color{white}\textbf{hg}}](a3) at (v4) {$a_4$};
        \node[agentstuck](a4) at (v8) {$a_3$};
        \node[agent](a5) at (v2) {$a_1$};
        \node[agent](a6) at (v1) {$a_2$};
        \draw[move] (a0)--(a1);
        \draw[move] (a1)--(a2);
        \draw[move] (a2)--(a3);
        \draw[move] (a3)--(a4);
        \draw[move] (a2)--(a5);
        \draw[move] (a5)--(a6);
        \node[smalldest] (d6) at (v5) {};
        \draw[move, dotted] (a6)--(d6);
      \end{tikzpicture}
      \subcaption{as DFST}
      \label{fig:pibt:tree}
    \end{minipage}
  \end{tabular}
  \caption{Example of PIBT.
    Requests for the next timestep are depicted by dashed circles, determined greedily according to agents' destinations (omitted here).
    Flows of priority inheritance and backtracking are drawn as single-line and doubled-line arrows, respectively.
    First, $a_7$ (blue agent) determines the next desired node (current location of $a_6$).
    Then, priority inheritance happens from $a_7$ to $a_6$, making $a_7$ wait for backtracking and $a_6$ start planning;
    $a_6$, $a_5$ and $a_4$ do the same.
    $a_3$ (red), however, is stuck (\ref{fig:pibt:init}).
    Thus, $a_3$ backtracks as invalid to $a_4$ (\ref{fig:pibt:btpi}).
    $a_4$ tries to replan, however, $a_4$ is also stuck hence $a_4$ sends backtracking as invalid to $a_5$  (blue with diagonal lines).
    $a_5$, with success replanning, executes other priority inheritance to $a_1$ (\ref{fig:pibt:btpi}).
    Finally, $a_1, a_5, a_6$ and $a_7$ receives backtracking as valid (\ref{fig:pibt:bt}) and then start moving (\ref{fig:pibt:result}).
    Fig.~\ref{fig:pibt:tree}: virtual depth first search tree for this example, for explanation of \algoname.
    $a_7$ is a root.
    Solid arrows are drawn from a parent to a child.
    $a_2$ finds an empty node.
  }
 \label{fig:pibt}
\end{figure*}
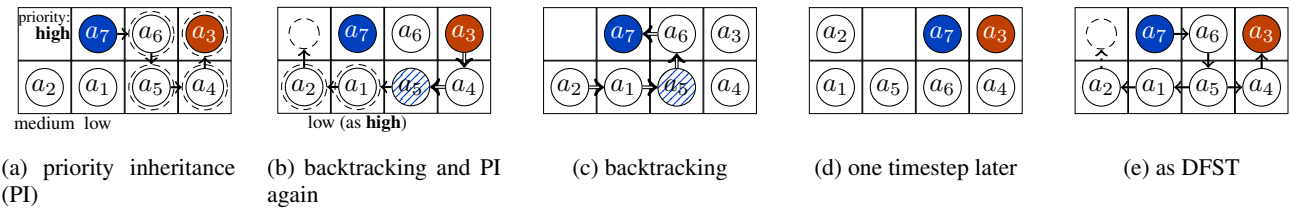
}

The main contributions of this paper are:
1)~the formalization of the time-independent model and \algoname, a proposed time-independent planning with guaranteed \emph{reachability}, i.e., all agents are ensured to reach their destinations within finite time.
\algoname, as a proof-of-concept, extends a recently-developed decoupled approach that solves MAPF iteratively, Priority Inheritance with Backtracking (PIBT)~\cite{okumura2019priority}.
We also present how an offline MAPF plan enhances \algoname.
2)~experimental results demonstrating the validity and robustness of the proposal through the simulation with stochastic delays of agents’ moves, using MAPF-DP (with Delay Probabilities)~\cite{ma2017multi}.

The remainder of this paper consists of the following five sections:
1) preliminary including the formalization of MAPF and related work,
2) the time-independent model,
3) examples of time-independent planning,
4) empirical results of the proposals using MAPF-DP, and
5) conclusion and future directions.

\section{Preliminaries}
\label{sec:preliminary}
This section first defines MAPF.
Then, we explain the MAPF variant emulating asynchrony of movements, called MAPF-DP, which we later use in experiments.
We also explain two policies that execute MAPF plans and PIBT, the original form of \algoname.

\subsection{MAPF}
In an environment represented as a graph $G = (V, E)$, the MAPF problem is defined as follows.
Let $\loc{i}{t} \in V$ denote the location of an agent $a_i$ at discrete time~$t \in \mathbb{N}$.
Given distinct initial locations $\loc{i}{0}\in V$ and destinations $g_i \in V$ for each agent, assign a path $\path{i} = (\loc{i}{0}, \loc{i}{1}, \dots, \loc{i}{T})$ such that $\loc{i}{T}=g_i$ to each agent minimizing some objective function (see below).

At each timestep $t$, $a_i$ can move to an adjacent node, or, can stay at its current location, i.e., $\loc{i}{t+1} \in \Neigh{\loc{i}{t}} \cup \{ \loc{i}{t} \}$, where \Neigh{v} is the set of nodes neighbor to $v \in V$.
Agents must avoid two types of conflicts~\cite{stern2019def}:
1) \emph{vertex conflict}: $\loc{i}{t} \neq \loc{j}{t}$, and,
2) \emph{swap conflict}: $\loc{i}{t} \neq \loc{j}{t+1} \lor \loc{i}{t+1} \neq \loc{j}{t}$.

Two kinds of objective functions are commonly used to evaluate MAPF solutions:
1)~\emph{sum of cost} (SOC), where the cost is the earliest timestep $T_i$ such that $\loc{i}{T_i}=g_i, \dots,\loc{i}{T}=g_i, T_i \leq T$,
2)~\emph{makespan}, i.e., $T$.
This paper focuses on the sum of cost (SOC).


\subsection{MAPF-DP (with Delay Probabilities)}
MAPF-DP~\cite{ma2017multi} emulates imperfect executions of MAPF plans by introducing the possibility of unsuccessful moves.
Time is still discrete.
At each timestep, an agent $a_i$ can either stay in place or move to an adjacent node with a probability $p_i$ of being unsuccessful.
The definition of conflicts is more restrictive than with normal MAPF:
1) vertex conflict is as defined in MAPF, and
2) \emph{following conflict}: $\loc{i}{t+1} \neq \loc{j}{t}$.
The rationale is that, without the later restriction, two agents might be in the same node due to one failing to move.
Note that following conflict contains swap conflict.

\subsubsection{Execution Policies}
Ma \etal~\cite{ma2017multi} studied two robust execution policies using MAPF plans for the MAPF-DP setting.
The first one, called \emph{Fully Synchronized Policies (FSPs)},
synchronizes the movements of agents globally, i.e., $a_i$ waits to move to \loc{i}{t+1} ($\neq \loc{i}{t}$) until all move actions of \loc{j}{t^\prime}, $t^\prime \leq t$ are completed.
The second approach, called \emph{Minimal Communication Policies (MCPs)}, executes a plan while maintaining its temporal dependencies.
There are two kinds of dependencies:
1) \emph{internal events}, i.e., the corresponding action of \loc{i}{t} is executed prior to that of \loc{i}{t+1}, and
2) \emph{node-related events}, i.e., if $\loc{i}{t} = \loc{j}{t^\prime}$ and $t < t^\prime$, the event of \loc{i}{t} is executed prior to that of \loc{j}{t^\prime}.
As long as an MAPF plan is valid, both policies make agents reach their destinations without conflicts, despite of delay probabilities.

\begin{figure*}[t]
 \centering
 \begin{tikzpicture}
  \node[vertex](v3) at (0, 0) {};
  \node[]() at (0.23, 0.25) {\scriptsize $v_3$};
  \node[vertex](v1) at (-\vertexsize, 0) {};
  \node[]() at (-0.45, 0.25) {\scriptsize $v_1$ };
  \node[vertex](v2) at (0, \vertexsize) {};
  \node[]() at (0.23, 0.95) {\scriptsize $v_2$ };
  \node[vertex](v4) at (0, -\vertexsize) {};
  \node[]() at (0.23, -0.45) {\scriptsize $v_4$ };
  \node[vertex](v5) at (\vertexsize, 0) {};
  \node[]() at (0.93, 0.25) {\scriptsize $v_5$ };
  \node[agenthighest](a1) at (v1) {$a_1$};
  \node[agentgreen](a2) at (v2) {$a_2$};
  \draw[destarrow,draw={rgb:red,0;green,1;blue,3}] (a1) -- (\vertexsize,0);
  \draw[destarrow,draw={rgb:green,1;blue,1}] (a2) -- (0,-\vertexsize);
  \node[anchor=west] at (1.2, 0) {
    \setlength{\tabcolsep}{2.55mm}
    \small
    \begin{tabular}{r||cc|c:c:cc:c:c:c:c:c:c:c:c:c:c}
      \toprule
      state
      & $\sigma_1$ & $\sigma_2$  
      & $\sigma_1$ & $\sigma_2$
      & $\sigma_1$ & $\sigma_2$  
      & $\sigma_2$ & $\sigma_1$ & $\sigma_1$ & $\sigma_1$ & $\sigma_1$ & $\sigma_2$ & $\sigma_2$
      & $\sigma_2$ & $\sigma_2$ & $\sigma_2$\\
      \midrule
      \mode{}& \mc & \mc & \mr & \mr & \me & \mc & \mr & \mc & \mr & \me & \mc & \me & \mc & \mr & \me & \mc \\
      \head{}& $\bot$ & $\bot$ & $v_3$ & $v_3$ & $v_3$ & $\bot$ & $v_3$ & $\bot$ & $v_5$ & $v_5$ & $\bot$ & $v_3$ & $\bot$ & $v_4$ & $v_4$ & $\bot$ \\
      \tail{}& $v_1$ & $v_2$ & $v_1$ & $v_2$ & $v_1$ & $v_2$ & $v_2$ & $v_3$ & $v_3$ & $v_3$ & $v_5$ & $v_2$ & $v_3$ & $v_3$ & $v_3$ & $v_4$ \\
      \bottomrule
    \end{tabular}
  };
  \node[anchor=east]() at (5.5, 1.1) {\small activated agent};
  \node[anchor=west]() at (7.7, 1.1) {\small interacted agent};
  \coordinate[](activated) at (5.5, 1.1);
  \coordinate[](s1-3) at (6.25, 1.0);
  \coordinate[](s2-2) at (5.6, 0.95);
  \draw[->] (activated) -- (s1-3);
  \draw[->] (activated) -- (s2-2);
  \coordinate[](interacted) at (7.75, 1.1);
  \coordinate[](s2-3) at (7.4, 0.95);
  \draw[->] (interacted) -- (s2-3);
  \node[anchor=north] at (7.0,-0.82) {\small interaction};
  \node[anchor=north] at (3.6,-0.82) {\small initial states};
  \node[anchor=south] at (13.9,0.8) {\small $\mc$: \contracted, $\mr$: \requesting, $\me$: \extended};
\end{tikzpicture}
 \caption{
   Example of an execution of the model.
   $a_1$ goes to $v_5$ and $a_2$ goes to $v_4$.
   In the table, time progresses from left to right.
   There is an interaction, which makes $a_2$ be back to \contracted.
 }
 \label{fig:execution}
\end{figure*}
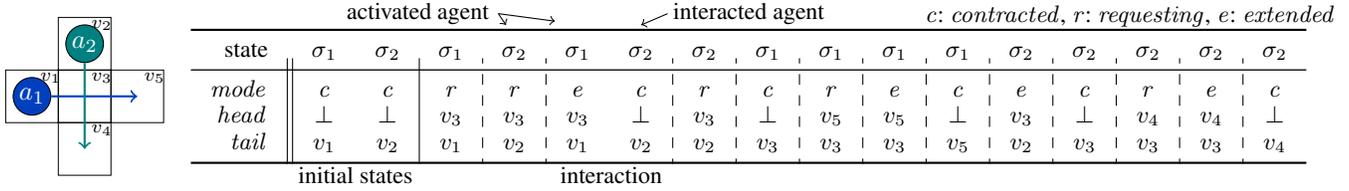

\subsection{Priority Inheritance with Backtracking (PIBT)}
PIBT~\cite{okumura2019priority} repeats one-timestep planning until it terminates, aiming at solving MAPF iteratively (without delays).
Unlike complete solvers for MAPF, PIBT does not ensure that all agents are on their goals simultaneously, rather it ensures \emph{reachability}; all agents are ensured to reach their destinations eventually.
Thus, an agent that reaches its goal potentially moves from there while other agents are moving.
Reachability plays a crucial role in situations where destinations are given continuously such as lifelong MAPF~\cite{ma2017lifelong} because all tasks (i.e., destinations) assigned to agents are ensured to be completed.

In PIBT, for each timestep, agents are provided with unique priorities and they sequentially determine their next locations in decreasing order of priorities while avoiding to use nodes that have requested by higher-priority agents.
\emph{Priority inheritance}, originally considered in resource scheduling problems~\cite{sha1990priority}, is introduced;
when a low-priority agent~$X$ obstructs the movement of a higher-priority agent~$Y$, $X$ temporarily inherits the higher-priority of $Y$.
Priority inheritance can be applied iteratively.
Agents giving their priority ($Y$) must wait for \textit{backtracking} from agents inheriting the priority ($X$).
Backtracking has two outcomes: valid or invalid.
A \emph{valid} situation occurs when $Y$ can successfully move to the location of $X$ in the next timestep.
An \emph{invalid} situation occurs when $X$ is stuck, i.e., neither going anywhere nor staying at its current location without colliding, forcing $Y$ to replan its path.
Fig.~\ref{fig:pibt} shows an example of PIBT in one timestep.

With these protocols, the agent with highest priority is ensured to move to an arbitrary neighbor node if the graph satisfies the adequate property, e.g., biconnected.
Subsequently, such an agent moves to its goal along the shortest path to its goal.
To ensure reachability, PIBT uses dynamic priorities, where the priority of an agent increments gradually at each timestep until it drops upon reaching its goal, meaning that, agents not reaching their goals yet eventually get the highest.

\section{Time-Independent Model}
\label{sec:system}
The time-independent model and the terminology used in this paper are partly inspired by the model for distributed algorithms with synchronous message passing~\cite{tel2000introduction} and the Amoebot model~\cite{derakhshandeh2014brief}, an abstract computational model for programmable matter.
Our model is however different in that agents are physically embodied and move on a graph, have destinations, and know their locations.

\paragraph{Components}
The system consists of a set of agents $A = \{a_1, \dots, a_n \}$ and a graph $G=(V,E)$.
We assume that each agent knows $G$ to plan their respective paths.

\paragraph{Configuration and State}
The whole system is represented as a transition system according to atomic actions of agents.
Each agent $a_i$ is itself a transition system with its \emph{state}, denoted as $\sigma_i$, consisting of its internal variables, its current location, and destination.
A \emph{configuration} $\gamma = (\sigma_1, \dots, \sigma_n)$ of the whole system at a given time consists of the states of all agents at that time.
A change in the state of some agents causes a change of configuration of the system, e.g., $\gamma = (\dots, \sigma_{i-1}, \sigma_i, \sigma_{i+1}, \dots) \rightarrow \gamma^\prime = (\dots, \sigma_{i-1}, \sigma_i^{\prime}, \sigma_{i+1},\dots)$.

\paragraph{Mode}
Each agent at any time occupies at least one node and occupies two nodes during moves between nodes.
We use two terms to represent agents' locations:
$\tail{i} \in V$ and $\head{i} \in V \cup \{ \bot \}$, where $\bot$ is void.
They are associated with a mode \mode{i} which can be: \contracted, \requesting, and \extended.
\begin{itemize}
  \item \contracted: $a_i$ stays at one node \tail{i}, and $\head{i} = \bot$; location of $a_i$ is \tail{i}
  \item \requesting: $a_i$ attempts to move to $\head{i}\not=\bot$, being at \tail{i}; location of $a_i$ is \tail{i}
  \item \extended: $a_i$ is moving from \tail{i} to \head{i}; locations of $a_i$ are both \tail{i} and \head{i}.
\end{itemize}
Initially, all agents are \contracted and have distinct initial locations.
\head{i} ($\neq \bot$) is always adjacent to \tail{i}.

\paragraph{Transition}
Agents move by changing modes.
These transitions are accompanied by changing \tail{} and \head{} as described below.
\begin{itemize}
  \item  from \contracted, $a_i$ can become \requesting by setting \head{i} to $u \in \Neigh{\tail{i}}$, to move to a neighbor node.
  \item from \requesting, $a_i$ can revert to \contracted by changing its \head{i} to $\bot$.
  \item from \requesting, agent $a_i$ can become \extended when $\lnot \occupied{\head{i}}$,
    where $\occupied{v}$ holds when there is no agent $a_j$ such that $\tail{j} = v$ and no agent $a_j$ in \extended such that $\head{j} = v$.
  \item from \extended, $a_i$ can become \contracted, implying that the movement is finished.
    This comes with $\tail{i} \leftarrow \head{i}$, then $\head{i} \leftarrow \bot$.
\end{itemize}
Transitions are atomic so, e.g., if agents in \contracted become \extended through \requesting, at least two actions are required; see an example later.
Other modes transitions are disallowed, e.g., an agent in \contracted cannot become \extended directly.

\paragraph{Conflict-freedom and Deadlock}
In the above transition rules, the model implicitly prohibits vertex and following conflicts of MAPF.
Rather, the model is prone to \emph{deadlocks};
A set of agents $\{ a_k, \ldots, a_l \}$ are in a deadlock when all of them are \requesting and are in a cycle $\head{k} = \tail{k+1}$, $\cdots$, $\head{l} = \tail{k}$.

\paragraph{Activation}
Agents perform a single atomic action as a result of being \emph{activated}. The nature and outcome of the action depend on the state of the agent (e.g., mode transition, local variable update). Activations occur non-deterministically, and there is no synchronization between agents.
For simplicity, we assume that at most one agent is activated at any time. In other words, the simultaneous activation of two agents $a_i$ and $a_j$ results in a sequence
$(\dots, \sigma_i, \dots, \sigma_j, \dots) \rightarrow
(\dots, \sigma_i^\prime, \dots, \sigma_j, \dots) \rightarrow
(\dots, \sigma_i^\prime, \dots, \sigma_j^\prime, \dots)$.
Activations are supposed to be \emph{fair} in the sense that,
in any sufficiently long period, all agents must be activated at least once.

\paragraph{Interaction}
An activation may affect not only variables of the activated agent, but also affect nearby agents indirectly.
For instance, if two \requesting agents have the same $\mathit{head}$, one wins and becomes \extended whereas the other loses and becomes \contracted, atomically.
This type of activation is called an \emph{interaction}.
Interactions include activations such that the activated agents change their variables referring to the variables of other agents.
We say that the agents involved in the interaction, except $a_i$ itself, are \emph{interacted agents}.
Given an activated agent $a_i$ and an interacted agent $a_j$, the system transitions as follows: $(\dots, \sigma_i, \dots, \sigma_j, \dots)  \rightarrow (\dots, \sigma_i^\prime, \dots, \sigma_j^\prime, \dots)$ with the state of all other agents unchanged.
Except for interactions, the configuration is changed by the state change of a single agent.
We assume that interactions are performed by communication between agents, but the detailed implementation is not relevant to this paper. 

\paragraph{Termination}
Assuming that each agent $a_i$ has its own destination $g_i \in V$, termination can be defined in two different ways.
\begin{itemize}
\item \emph{Strong termination} occurs when reaching a configuration such that $\mode{i} = \contracted \land \tail{i} = g_i$ for any agent~$a_i$.
\item \emph{Weak termination} is when all agents have been at least once in a state where $\mode{i} = \contracted \land \tail{i} = g_i$.
\end{itemize}
Strong termination corresponds to usual MAPF termination, whereas weak termination corresponds to the reachability property of PIBT. Strong termination implies weak termination.

In this paper, we refer to weak termination as \emph{reachability}.
Note that the properties of either deadlock-freedom or deadlock-recovery are required to ensure reachability.


\paragraph{Remarks}
Figure~\ref{fig:execution} illustrates an execution in the time-independent model.
Although the example focuses on (classical) MAPF, many other problem variants, e.g., iterative MAPF~\cite{okumura2019priority}, can be addressed simply by adapting termination and goal assignments.

\section{Algorithm}
\label{sec:algorithm}
This section presents two examples of time-independent planning: \greedy and \algoname.
We also present how to enhance \algoname with offline MAPF plans.

\subsection{Greedy Approach}
\greedy performs only basic actions;
it can be a template for another time-independent planning.
We simply describe its implementation for $a_i$ as follows.
\begin{itemize}
\item when \contracted:
  Choose the nearest node to $g_i$ from \Neigh{\tail{i}} as new \head{i}, then become \requesting.
\item when \requesting:
  Become \extended when the \head{i} is unoccupied, otherwise, do nothing.
\item when \extended:
  Become \contracted.
\end{itemize}
Obviously, \greedy is prone to deadlocks, e.g, when two adjacent agents try to swap their locations, they block eternally.
The time-independent planning without deadlock-freedom or deadlock-recovery properties is impractical, motivating the next algorithm.

\subsection{Causal-PIBT}
The \algoname{} algorithm extends both algorithms PIBT and \greedy.
Although PIBT is timing-based relying on synchronous moves, \algoname is event-based relying on causal dependencies of agents' actions.

\subsubsection{Concept}
Two intuitions are obtained from PIBT to design a time-independent algorithm that ensures reachability;
1)~Build a depth-first search tree rooted at the agent with highest priority, using the mechanism of priority inheritance and backtracking.
2)~Drop priorities of agents that arrive at their goals to give higher priorities to all agents not having reached their goals yet, thus resulting in that all agents eventually reach their goals.
We complement the first part as follows.

The path adjustment of PIBT in one timestep can be seen as the construction of a (virtual) depth-first search tree consisting of agents and their dependencies.
The tree is rooted at the first agent starting priority inheritance, i.e., locally highest priority agent, e.g., $a_7$ in Fig.~\ref{fig:pibt}.
When an agent $a_j$ inherits a priority from another agent $a_i$, $a_j$ becomes a \emph{child} of $a_i$ and $a_i$ becomes its \emph{parent}.
We show this virtual tree in Fig.~\ref{fig:pibt:tree}.
Once an empty node adjacent to the tree is found, all agents on the path from the root to that empty node can move toward one step forward
($a_2, a_1, a_5, a_6$ and $a_7$).
This enables the agent with highest priority, being always a root, to move from the current node to any arbitrary neighbor nodes.
The invalid outcome in backtracking works as backtracking in a depth-first search,
while the valid outcome notifies the search termination.

\begin{algorithm}[ht!]
 \caption{\algoname}
 \label{algo:causal}
 {\small
 \begin{algorithmic}
  \State $\parent{i} \in A$: initially $a_i$; $\children{i} \subset A$: initially $\emptyset$
  \State $\pori{i}$: original priority; $\ptmp{i}$: temporal priority, initially \pori{i}
  \State $C_i \subseteq V$: candidate nodes, initially $\Neigh{\tail{i}} \cup \{ \tail{i} \}$
  \State $S_i \subseteq V$: searched nodes, initially $\emptyset$
  \vspace{0.05cm}
 \end{algorithmic}
 \begin{algorithmic}[1]
  \When{$\mode{i} = \contracted$}
  \If{$C_i = \emptyset \land \parent{i} = a_i$}
  \label{algo:causal:detect-searchend}
  \State \Call{ReleaseChildren}{}, \Call{Reset}{}
  \EndIf
  \label{algo:causal:detect-searchend-end}
  \State \Call{PriorityInheritance}{}
  \label{algo:causal:pi-cont}
  \If{$C_i = \emptyset$}
  \label{algo:causal:detect-empty}
  \State let $a_j$ be \parent{i}
  \If{$\head{j} = \tail{i}$}
  \label{algo:causal:bt-start}
  \State $S_j \leftarrow S_j \cup S_i$, $C_j \leftarrow C_j \setminus S_j$
  \label{algo:causal:backprop}
  \Comment with $a_j$
  \State $\head{j} \leftarrow \bot$, $\mode{j} \leftarrow \contracted$
  \EndIf
  \label{algo:causal:bt-end}
  \State \Return
  \EndIf
  \label{algo:causal:detect-empty-end}
  \State $u \leftarrow$~the nearest node to $g_i$ in $C_i$
  \label{algo:causal:node-select}
  \If{$u = \tail{i}$}
  \State \Call{ReleaseChildren}{}, \Call{Reset}{}
  \State \Return
  \EndIf
  \State $C_i \leftarrow C_i \setminus \{ u \}$, $S_i \leftarrow S_i \cup \{ u, \tail{i} \}$
  \label{algo:causal:update-cs}
  \State $\head{i} \leftarrow u$, $\mode{i} \leftarrow \requesting$
  \label{algo:causal:goto-requesting}
  \EndWhen
  \vspace{0.05cm}
  \When{$\mode{i} = \requesting$}
  \State \Call{PriorityInheritance}{}
  \label{algo:causal:pi-ext}
  \If{$\parent{i} \neq a_i \land \head{i} \in S_{\parent{i}}$}
  \label{algo:causal:backcont-s}
  \Comment{\parent{i}}
  \State $\head{i} \leftarrow \bot$, $\mode{i} \leftarrow \contracted$
  \State \Return
  \EndIf
  \label{algo:causal:backcont-e}
  \Ifsingle{$\occupied{\head{i}}$}{\Return}
  \label{algo:causal:pass-ext}
  \State {\scriptsize $A^\prime \leftarrow \{ a_j \;|\; a_j \in A,
  \mode{j} = \requesting, \head{j} = \head{i}\}$}
  \State $a^\star \leftarrow \argmax_{a_j \in A^\prime}\mathit{ptmp}_{j}$
  \For{$a_j \in A^\prime \setminus \{ a^\star \}$}
  \Comment agents in $A^\prime$
  \State $\head{j} \leftarrow \bot$, $\mode{j} \leftarrow \contracted$
  \EndFor
  \label{algo:causal:winner-determination-end}
  \Ifsingle{$a^\star \neq a_i$}{\Return}
  \label{algo:causal:if-winner}
  \State {\footnotesize $\children{\parent{i}} \leftarrow \children{\parent{i}} \setminus \{ a_i \}$}
  \Comment \parent{i}
  \State $\parent{i} \leftarrow a_i$
  \State \Call{ReleaseChildren}{}
  \State $\mode{i} \leftarrow extended$
  \label{algo:causal:if-winner-end}
  \EndWhen
  \vspace{0.05cm}
  \When{$\mode{i} = \extended$}
  \State $\tail{i} \leftarrow \head{i}$, $\head{i} \leftarrow \bot$, $\mode{i} \leftarrow \contracted$
  \State update \pori{i}, \Call{Reset}{}
  \label{algo:causal:updatepori}
  \EndWhen
\end{algorithmic}
}
\end{algorithm}

\subsubsection{Description}
\algoname basically performs \greedy, i.e., at each activation an agent $a_i$ tries to move towards its goal; in addition, \algoname uses priority inheritance.
When $a_i$ is \contracted or \requesting, $a_i$ potentially inherits priority from another agent $a_j$ in \requesting with $\head{j} = \tail{i}$.
We now explain the details.
The pseudocode is split into Algorithm~\ref{algo:causal} and~\ref{algo:procs}.
Interactions and the corresponding interacted agents are explicitly marked in comments.
In Algo.~\ref{algo:causal}, procedures with activation are denoted for each mode.

\smallskip
\noindent
\underline{Variants:}
We first introduce local variants of $a_i$.
\begin{itemize}
  \item \parent{i} and \children{i}: for maintaining tree structures.
    $a_i$ is a root when $\parent{i} = a_i$.
    The algorithm updates these variants so that $a_i = \parent{j} \Leftrightarrow a_j \in \children{i}$.
  \item $C_i$ and $S_i$: for searching unoccupied neighbor nodes.
    $C_i$ is candidate nodes of next locations.
    $a_i$ in \contracted selects \head{i} from $C_i$.
    $S_i$ represents already searched nodes by a tree to which $a_i$ belongs.
    $S_i$ is propagated in the tree.
    $C_i$ is updated to be disjoint from $S_i$.
  \item \pori{i} and \ptmp{i}: priorities.
    They are components of a total order set.
    \ptmp{i} is basically equal to \pori{i}, however, it is changed by priority inheritance.
    $\ptmp{i} \geq \pori{i}$ in any time, and only \ptmp{i} is used for interaction.
\end{itemize}

\smallskip
\noindent
\underline{Structure:}
The procedures in \contracted consist of:
\begin{itemize}
  \item Relaying release [Line~\ref{algo:causal:detect-searchend}--\ref{algo:causal:detect-searchend-end}].
    An agent in stuck like $a_4$ in Fig.~\ref{fig:pibt:tree} is ``released'' when its parent ($a_5$) moves.
    If it has children ($a_3$), it needs to relay the release to its children, then initialize.
  \item Priority inheritance [Line~\ref{algo:causal:pi-cont}].
    When $a_i$ is activated and there exists an agent $a_j$ with higher priority such that $\head{j} = \tail{i}$, then priority inheritance happens from $a_j$ to $a_i$, e.g., from $a_7$ and $a_6$ in Fig.~\ref{fig:pibt:init}.
  \item Backtracking [Line~\ref{algo:causal:detect-empty}--\ref{algo:causal:detect-empty-end}].
    When $a_i$ detects stuck ($C_i = \emptyset$), making its parent $a_j$ \contracted while propagating already searched nodes $S_i$, then stops.
    This procedures correspond to invalid case of backtracking in PIBT, e.g., from $a_3$ to $a_4$ in Fig.~\ref{fig:pibt:btpi}.
  \item Prioritized planning [Line~\ref{algo:causal:node-select}--\ref{algo:causal:goto-requesting}].
    Pickup one node $u$ for the next location from $C_i$, update searched nodes, and become \requesting.
    Initialize when $u$ is \tail{i}.
\end{itemize}
The procedures in \requesting consist of:
\begin{itemize}
  \item Priority inheritance [Line~\ref{algo:causal:pi-ext}].
    The condition is the same as priority inheritance in \contracted.
    This implicitly contributes to detecting deadlocks.
  \item Deadlock detection and resolving [Line~\ref{algo:causal:backcont-s}--\ref{algo:causal:backcont-e}].
    This part detects the circle of requests, then go back to \contracted.
  \item Winner determination [Line~\ref{algo:causal:pass-ext}--\ref{algo:causal:winner-determination-end}].
    When \head{i} is unoccupied, there is a chance to move there.
    First, identify agents requesting the same node then select one with highest priority as a winner.
    Let losers back to be contracted.
    See ``interaction'' in Fig.~\ref{fig:execution}.
  \item Preparation for moves [Line~\ref{algo:causal:if-winner}--\ref{algo:causal:if-winner-end}].
    If $a_i$ itself is a winner, after isolating itself from a tree belonged to, $a_i$ becomes \extended.
\end{itemize}
The procedures in \extended are to update priority and back to \contracted.

\smallskip
\noindent
\underline{Subprocedures:}
We use three subprocedures, as shown in Algo.~\ref{algo:procs}.
\textsc{PriorityInheritance} first determines whether priority inheritance should occur [Line~\ref{algo:procs:check-start}--\ref{algo:procs:check-end}], then updates the tree structure and inherits both the priority and the searched nodes from the new parent [Line~\ref{algo:procs:update-tree}--\ref{algo:procs:end-inheriting}].
\textsc{ReleaseChildren} just cuts off the relationship with the children of $a_i$.
\textsc{Reset} initializes the agent's status.

\smallskip
\noindent
\underline{Overview:}
\algoname constructs depth-first search trees, each rooted at the agents with locally highest priorities, using \parent{} and \children{}.
They are updated through priority inheritance and backtracking mechanisms.
All agents in the same tree have the same priority \ptmp{}, equal to the priority \pori{} of the rooted agent.
When a tree with higher priority comes in contact with a lower-priority tree (by some agent becoming \requesting), the latter tree is decomposed and is partly merged into the former tree (by \textsc{PriorityInheritance}).
In the reverse case, the tree with lower priority just waits for the release of the area.
Backtracking occurs when a child has no candidate node to move due to requests from agents with higher priorities [Line~\ref{algo:causal:bt-start}--\ref{algo:causal:bt-end} in Algo.~\ref{algo:causal}].

\smallskip
\noindent
\underline{Dynamic Priorities:}
We assume that \pori{i} is unique between agents in any configuration and updated upon reaching a goal to be lower than priorities of all agents who have not yet reached their goals [Line~\ref{algo:causal:updatepori} in Algo.~\ref{algo:causal}].
This is realized by a prioritization scheme similar to PIBT.

\paragraph{Properties}
Next, two useful properties of \algoname are shown: deadlock-recovery and the reachability.
We also state that \algoname is quick enough as atomic actions in terms of time complexity.
\begin{theorem}[deadlock-recovery]
 \label{theorem:deadlock-recovery}
 \algoname ensures no deadlock situation can last forever.
\end{theorem}
\begin{proof}
 Assume that there is a deadlock.
 When an agent $a_i$ changes from \contracted to \requesting, it updates $S_i$ such that $S_i$ includes \head{i} and \tail{i} [Line~\ref{algo:causal:update-cs} in Algo.~\ref{algo:causal}].
 After finite activations, all agents involved in the deadlock must have the same priority \ptmp{} due to priority inheritance.
 Every priority inheritance is accompanied by the propagation of $S_i$ [Line~\ref{algo:procs:end-inheriting} in Algo.~\ref{algo:procs}].
 When an agent in \requesting detects its head  in $S_{\parent{i}}$, it changes back to \contracted [Line~\ref{algo:causal:backcont-s}--\ref{algo:causal:backcont-e} in Algo.~\ref{algo:causal}].
 These imply the statement.
\end{proof}

Once an agent $a_i$ in a deadlock comebacks to \contracted, $a_i$ cannot request the same place as before due to the update of $C_i$ [Line~\ref{algo:causal:update-cs} in Algo.~\ref{algo:causal}], preventing a reoccurrence of the deadlock situation; i.e., a livelock situation if repeated indefinitely.

\begin{theorem}[reachability]
 \label{theorem:reachability}
 \algoname has the reachability in biconnected graphs if $|A| < |V|$.
\end{theorem}
\begin{sketch}
  Let $a_i$ be the agent with highest priority \pori{}.
  If $a_i$ is \contracted, this agent does not inherit any priority.
  Thus, $a_i$ can be \requesting so that \head{i} is any neighbor node of \tail{i}.
  If $a_i$ is \requesting, $a_i$ eventually moves to \head{i} due to the construction of the depth-first search tree in a biconnected graph.
  These imply that $a_i$ can move to an arbitrary neighbor node in finite activations.
  By this, $a_i$ moves along the shortest path to its goal.
  Due to the prioritization scheme, an agent that has not reached its goal eventually inherits the highest priority \pori{}, and then starts moving along its shortest path to its goal. This ensures reachability.
  The detailed proof is found in the long version~\cite{okumura2020causalpibt:techrep}.
\end{sketch}

\begin{algorithm}[t]
 \caption{Procedures of \algoname}
 \label{algo:procs}
 {\small
 \begin{algorithmic}[1]
  \Procedure{PriorityInheritance}{}
  \State {\scriptsize $A^{\prime\prime} \leftarrow \{ a_j \;|\; a_j \in A,
  \mode{j} = \requesting, \tail{i} = \head{j} \}$ }
  \label{algo:procs:check-start}
  \Ifsingle{$A^{\prime\prime} = \emptyset$}{\Return}
  \State $a_k \leftarrow \argmax_{a_j \in A^{\prime\prime}} \ptmp{j}$
  \Ifsingle{$\ptmp{k} \leq \ptmp{i}$}{\Return}
  \label{algo:procs:check-end}
  \State \Call{ReleaseChildren}{}
  \label{algo:procs:update-tree}
  \State {\footnotesize $\children{\parent{i}} \leftarrow \children{\parent{i}} \setminus \{ a_i \}$}
  \Comment \parent{i}
  \State $\parent{i} \leftarrow a_k$, $\children{k} \leftarrow \children{k} \cup \{ a_i \}$
  \Comment $a_k$
  \State $\ptmp{i} \leftarrow \ptmp{k}$
  \State $S_i \leftarrow S_k \cup \{ \head{i} \}$, $C_i \leftarrow \Neigh{\tail{i}} \cup \{ \tail{i} \} \setminus S_i $
  \label{algo:procs:end-inheriting}
  \EndProcedure
  \vspace{0.05cm}
  \Procedure{ReleaseChildren}{}
  \Forsingle{$a_j \in \children{i}$}{$\parent{j} \leftarrow a_j$ \Comment $a_j$}
  \State $\children{i} \leftarrow \emptyset$
  \EndProcedure
  \vspace{0.05cm}
  \Procedure{Reset}{}
  \State $S_i \leftarrow \emptyset$, $C_i \leftarrow \Neigh{\tail{i}} \cup \{ \tail{i} \}$
  \State $\mathit{ptmp}_i \leftarrow \mathit{pori}_i$
  \label{algo:procs:updateptmp}
  \EndProcedure
\end{algorithmic}
}
\end{algorithm}

\begin{proposition}[time complexity]
  Assume that the maximum time required for one operation to update $S_i$ or $C_i$ (i.e., union, set minus) is $\alpha$ and for an operation of Line~\ref{algo:causal:node-select} in Algo.~\ref{algo:causal} is $\beta$.
  Let $\Delta(G)$ denote the maximum degree of $G$.
  For each activation, time complexity of \algoname in \contracted is $O(\Delta(G)+\alpha+\beta)$, in \requesting is $O(\Delta(G)+\alpha)$, and in \extended is $O(1)$.
\end{proposition}
\begin{sketch}
  $|\children{i}|$, $|A^\prime|$ used in \requesting, and $|A^{\prime\prime}|$ used in \textsc{PriorityInheritance} are smaller than equal to $\Delta(G)$.
  Then, the complexity of \textsc{PriorityInheritance} is $O(\Delta(G)+\alpha)$.
  \textsc{ReleaseChildren} is $O(\Delta(G))$.
  \textsc{Reset} is $O(1)$.
  We can derive the statements using the above.
\end{sketch}

{
  \newcommand{\colwidth}{1\hsize}
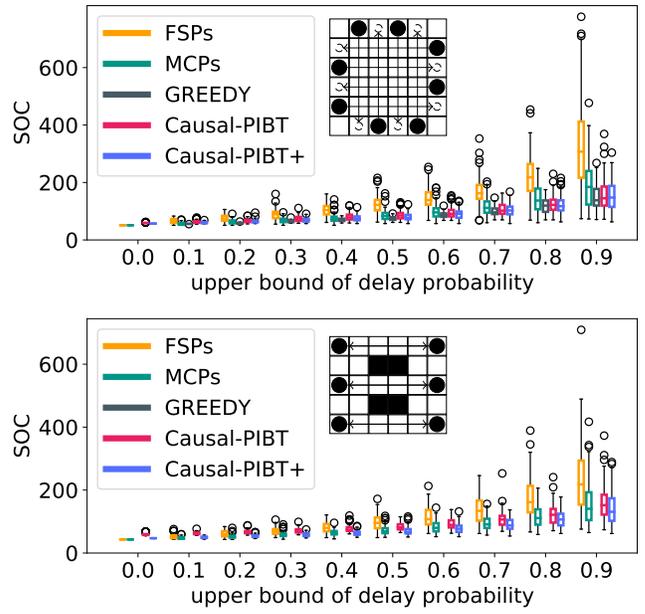
\begin{figure}[t]
  \centering
  \begin{tabular}{@{}c@{}}
    \begin{minipage}{\colwidth}
      \centering
      \begin{tikzpicture}
        \node[vertex-s](v0-0) at (4.4, 2.3) {};
      \node[vertex-s,right=0cm of v0-0](v1-0) {};
      \node[vertex-s,right=0cm of v1-0](v2-0) {};
      \node[vertex-s,right=0cm of v2-0](v3-0) {};
      \node[vertex-s,right=0cm of v3-0](v4-0) {};
      \node[vertex-s,right=0cm of v4-0](v5-0) {};
      \node[vertex-s,above=0cm of v0-0](v0-1) {};
      \node[vertex-s,above=0cm of v0-1](v0-2) {};
      \node[vertex-s,above=0cm of v0-2](v0-3) {};
      \node[vertex-s,above=0cm of v0-3](v0-4) {};
      \node[vertex-s,above=0cm of v0-4](v0-5) {};
      \node[vertex-s,above=0cm of v1-0](v1-1) {};
      \node[vertex-s,above=0cm of v1-1](v1-2) {};
      \node[vertex-s,above=0cm of v1-2](v1-3) {};
      \node[vertex-s,above=0cm of v1-3](v1-4) {};
      \node[vertex-s,above=0cm of v1-4](v1-5) {};
      \node[vertex-s,above=0cm of v2-0](v2-1) {};
      \node[vertex-s,above=0cm of v2-1](v2-2) {};
      \node[vertex-s,above=0cm of v2-2](v2-3) {};
      \node[vertex-s,above=0cm of v2-3](v2-4) {};
      \node[vertex-s,above=0cm of v2-4](v2-5) {};
      \node[vertex-s,above=0cm of v3-0](v3-1) {};
      \node[vertex-s,above=0cm of v3-1](v3-2) {};
      \node[vertex-s,above=0cm of v3-2](v3-3) {};
      \node[vertex-s,above=0cm of v3-3](v3-4) {};
      \node[vertex-s,above=0cm of v3-4](v3-5) {};
      \node[vertex-s,above=0cm of v4-0](v4-1) {};
      \node[vertex-s,above=0cm of v4-1](v4-2) {};
      \node[vertex-s,above=0cm of v4-2](v4-3) {};
      \node[vertex-s,above=0cm of v4-3](v4-4) {};
      \node[vertex-s,above=0cm of v4-4](v4-5) {};
      \node[vertex-s,above=0cm of v5-0](v5-1) {};
      \node[vertex-s,above=0cm of v5-1](v5-2) {};
      \node[vertex-s,above=0cm of v5-2](v5-3) {};
      \node[vertex-s,above=0cm of v5-3](v5-4) {};
      \node[vertex-s,above=0cm of v5-4](v5-5) {};
      \node[agent-s](a0) at (v0-3) {};
      \node[dest-s ](d0) at (v5-3) {};
      \node[agent-s](a1) at (v0-1) {};
      \node[dest-s ](d1) at (v5-1) {};
      \node[agent-s](a2) at (v2-0) {};
      \node[dest-s ](d2) at (v2-5) {};
      \node[agent-s](a3) at (v4-0) {};
      \node[dest-s ](d3) at (v4-5) {};
      \node[agent-s](a4) at (v5-2) {};
      \node[dest-s ](d4) at (v0-2) {};
      \node[agent-s](a5) at (v5-4) {};
      \node[dest-s ](d5) at (v0-4) {};
      \node[agent-s](a6) at (v3-5) {};
      \node[dest-s ](d6) at (v3-0) {};
      \node[agent-s](a7) at (v1-5) {};
      \node[dest-s ](d7) at (v1-0) {};
      \foreach \u / \v in {a0/d0,a1/d1,a2/d2,a3/d3,a4/d4,a5/d5,a6/d6,a7/d7}
      \draw[destarrow-s] (\u) -- (\v);
      \node[anchor=south west,inner sep=0] at (0,0)
      {\includegraphics[width=1\hsize]{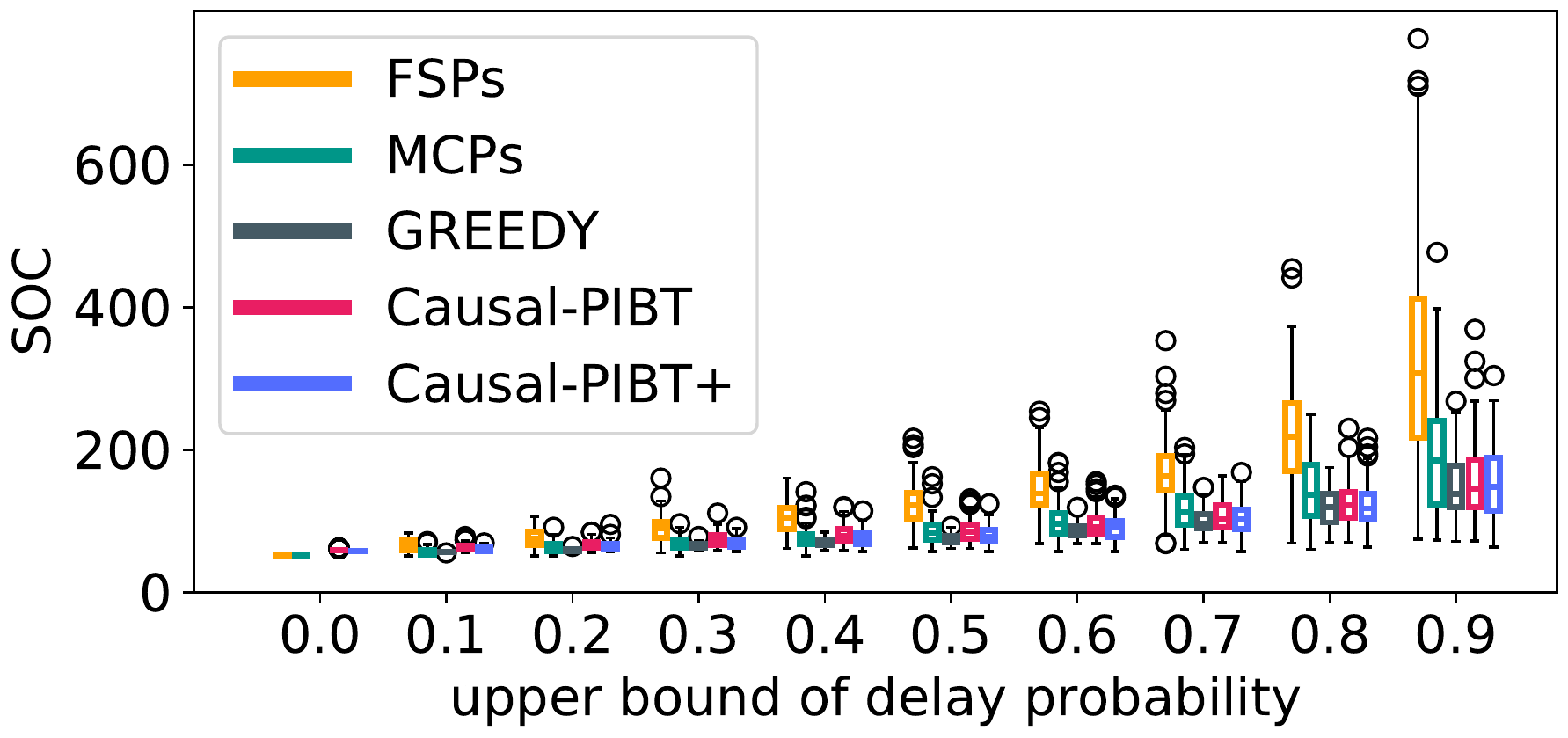}};
      \end{tikzpicture}
    \end{minipage}
    \\
    \begin{minipage}{\colwidth}
      \centering\medskip
      \begin{tikzpicture}
      \node[vertex-s](v0-0) at (4.4, 2.5) {};
      \node[vertex-s,right=0cm of v0-0](v1-0) {};
      \node[vertex-s,right=0cm of v1-0](v2-0) {};
      \node[vertex-s,right=0cm of v2-0](v3-0) {};
      \node[vertex-s,right=0cm of v3-0](v4-0) {};
      \node[vertex-s,right=0cm of v4-0](v5-0) {};
      \node[vertex-s,above=0cm of v0-0](v0-1) {};
      \node[vertex-s,above=0cm of v0-1](v0-2) {};
      \node[vertex-s,above=0cm of v0-2](v0-3) {};
      \node[vertex-s,above=0cm of v0-3](v0-4) {};
      \node[vertex-s,above=0cm of v1-0](v1-1) {};
      \node[vertex-s,above=0cm of v1-1](v1-2) {};
      \node[vertex-s,above=0cm of v1-2](v1-3) {};
      \node[vertex-s,above=0cm of v1-3](v1-4) {};
      \node[vertex-s-obj,above=0cm of v2-0](v2-1) {};
      \node[vertex-s,above=0cm of v2-1](v2-2) {};
      \node[vertex-s-obj,above=0cm of v2-2](v2-3) {};
      \node[vertex-s,above=0cm of v2-3](v2-4) {};
      \node[vertex-s-obj,above=0cm of v3-0](v3-1) {};
      \node[vertex-s,above=0cm of v3-1](v3-2) {};
      \node[vertex-s-obj,above=0cm of v3-2](v3-3) {};
      \node[vertex-s,above=0cm of v3-3](v3-4) {};
      \node[vertex-s,above=0cm of v4-0](v4-1) {};
      \node[vertex-s,above=0cm of v4-1](v4-2) {};
      \node[vertex-s,above=0cm of v4-2](v4-3) {};
      \node[vertex-s,above=0cm of v4-3](v4-4) {};
      \node[vertex-s,above=0cm of v5-0](v5-1) {};
      \node[vertex-s,above=0cm of v5-1](v5-2) {};
      \node[vertex-s,above=0cm of v5-2](v5-3) {};
      \node[vertex-s,above=0cm of v5-3](v5-4) {};
      \node[agent-s](a0) at (v0-0) {};
      \node[agent-s](a1) at (v0-2) {};
      \node[agent-s](a2) at (v0-4) {};
      \node[agent-s](a3) at (v5-0) {};
      \node[agent-s](a4) at (v5-2) {};
      \node[agent-s](a5) at (v5-4) {};
      \foreach \u / \v in {a0/v5-0,a1/v5-2,a2/v5-4,a3/v0-0,a4/v0-2,a5/v0-4}
      \draw[destarrow-s] (\u) -- (\v);
      \node[anchor=south west,inner sep=0] at (0,0)
      {\includegraphics[width=1\hsize]{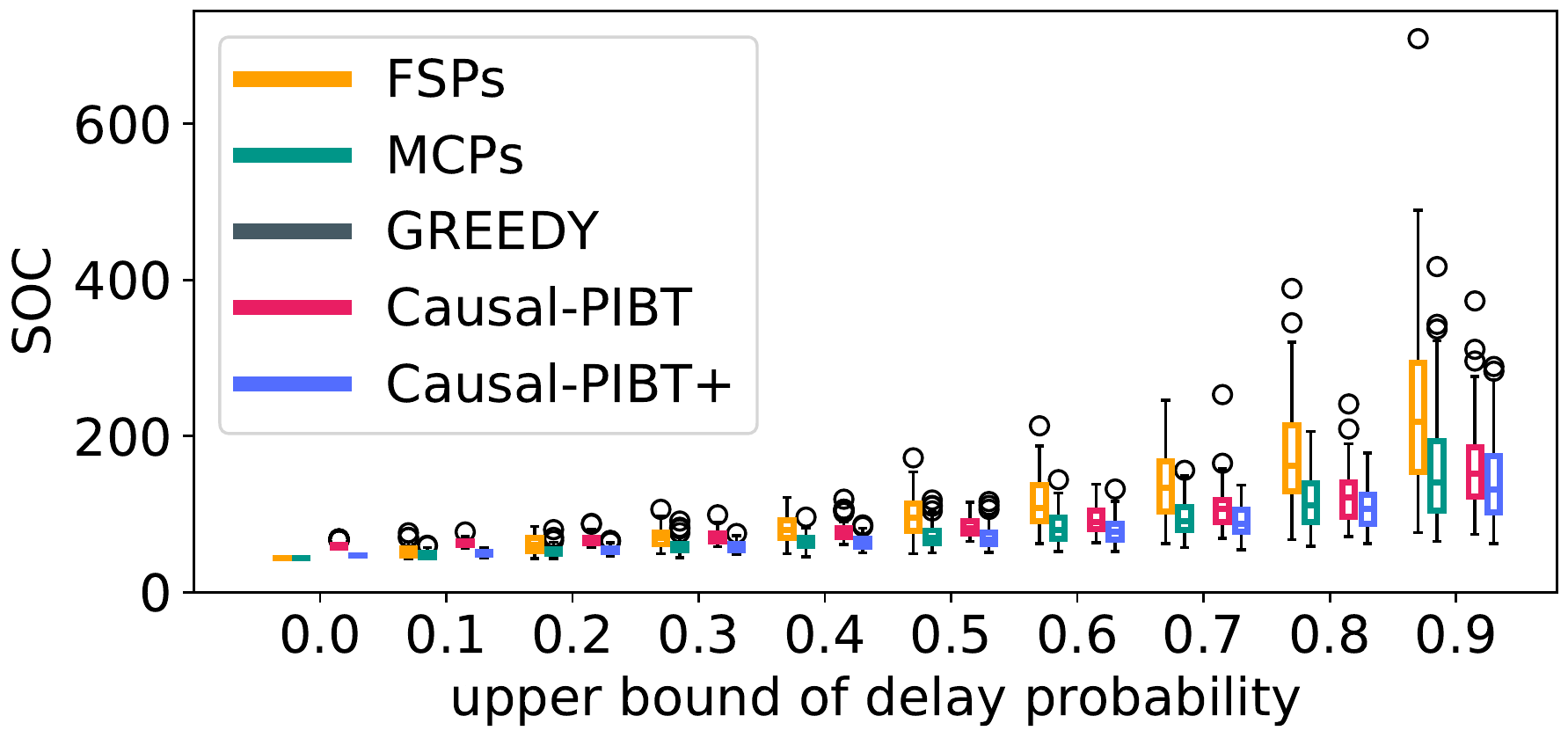}};
    \end{tikzpicture}
  \end{minipage}
  \end{tabular}
  \caption{The results in small benchmarks, shown as boxplots.
    The number of successful instances for \greedy are 0, 5, 8, 11, 17, 26, 30, 36, 40, 51 for each $\bar{p}$ (from $0$ to $0.9$) in the upper plot.
    \greedy failed all cases in the lower plot.
  }
  \label{fig:result:1}
\end{figure}
}

\subsection{MAPF Plans as Hints}
\label{subsec:heuristics}
Although \greedy and \algoname are for online situations, they can optionally use offline MAPF plans as hints while still being online and distributed schemes during execution.
This can contribute to relaxing congestion because time-independent planning is shortsighted, i.e., planning paths anticipating only a single step ahead.
Finding optimal MAPF plans is NP-hard~\cite{yu2013structure,ma2016multi}, however, powerful solvers have been developed so far~\cite{stern2019multi};
time-independent planning can use them.
We describe how \algoname is enhanced by MAPF plans.
Note that \greedy can be extended as well.
The intuition is to make agents follow original plans whenever possible.

Given a MAPF plan \paths, assume that $a_i$ knows \path{i} \apriori.
Before execution, $a_i$ makes a new path $\tilde{\pi}_i$ by removing a node \loc{i}{t} such that $\loc{i}{t} = \loc{i}{t-1}$ from \path{i} while keeping the order of nodes, since the action ``stay'' is meaningless in the time-independent model.
During execution, $a_i$ manages its own discrete time $t_i$ internally (initially $t_i = 0$).

The node selection phase [Line~\ref{algo:causal:node-select} in Algo.~\ref{algo:causal}] is modified as follows.
If $t_i \geq |\tilde{\pi}_i|-1$, i.e., $a_i$ finished $\tilde{\pi}_i$, $a_i$ follows the original policy; chooses a next node greedily from $C_i$.
Otherwise, if $\tail{i} = \loc{i}{t_i}$ and $\loc{i}{t_i + 1} \in C_i$, i.e., when $a_i$ can follow $\tilde{\pi}_i$, $a_i$ selects $\loc{i}{t_i + 1}$.
Or else, $a_i$ selects the nearest node on the rest of $\tilde{\pi}_i$, formally;
$\argmin_{v \in C_i} \left\{ \min_{u \in \tilde{\pi}[t_i + 1\colon]} \text{cost}(v, u) \right\}$ where $\tilde{\pi}[t\colon] = (\tilde{\pi}[t], \tilde{\pi}[t+1],\dots)$.

$t_i$ is updated when $a_i$ in \extended is activated.
If $\head{i} \in \tilde{\pi}_i[t_i+1\colon]$, $t_i$ becomes the corresponding timestep $t^\prime$ such that $\tilde{\pi}_i[t^\prime] = \head{i}$; otherwise, do nothing.

{
  \newcommand{\colwidth}{0.9\hsize}
\begin{figure}[t]
  \centering
  \begin{tabular}{@{}c@{}}
    \begin{minipage}{\colwidth}
      \centering
      \begin{tikzpicture}
        \node[anchor=south west,inner sep=0] at (0,0)
        {\includegraphics[width=1\hsize]{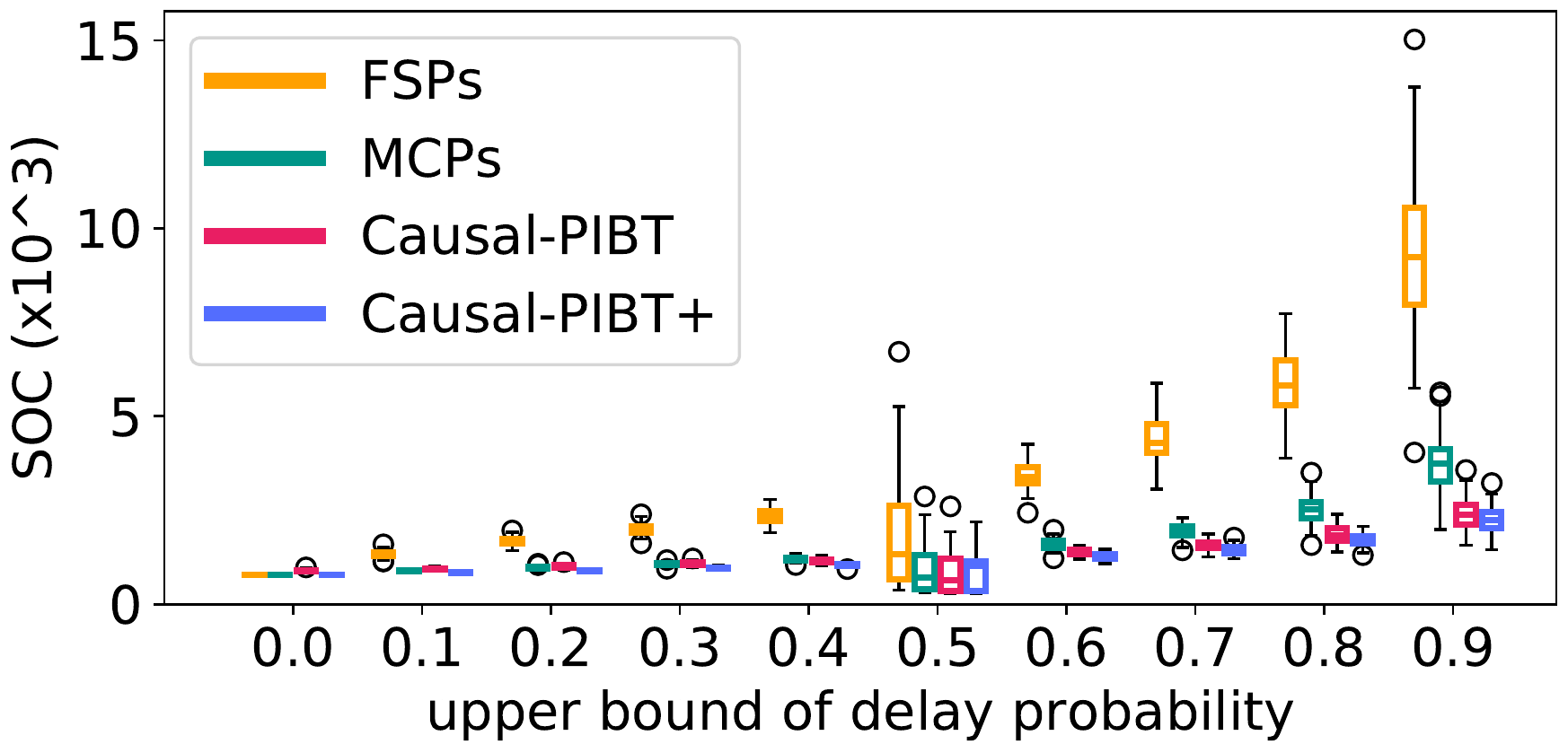}};
        \node[anchor=south west,inner sep=0] at (3.7,2.55)
        {\includegraphics[width=0.13\hsize]{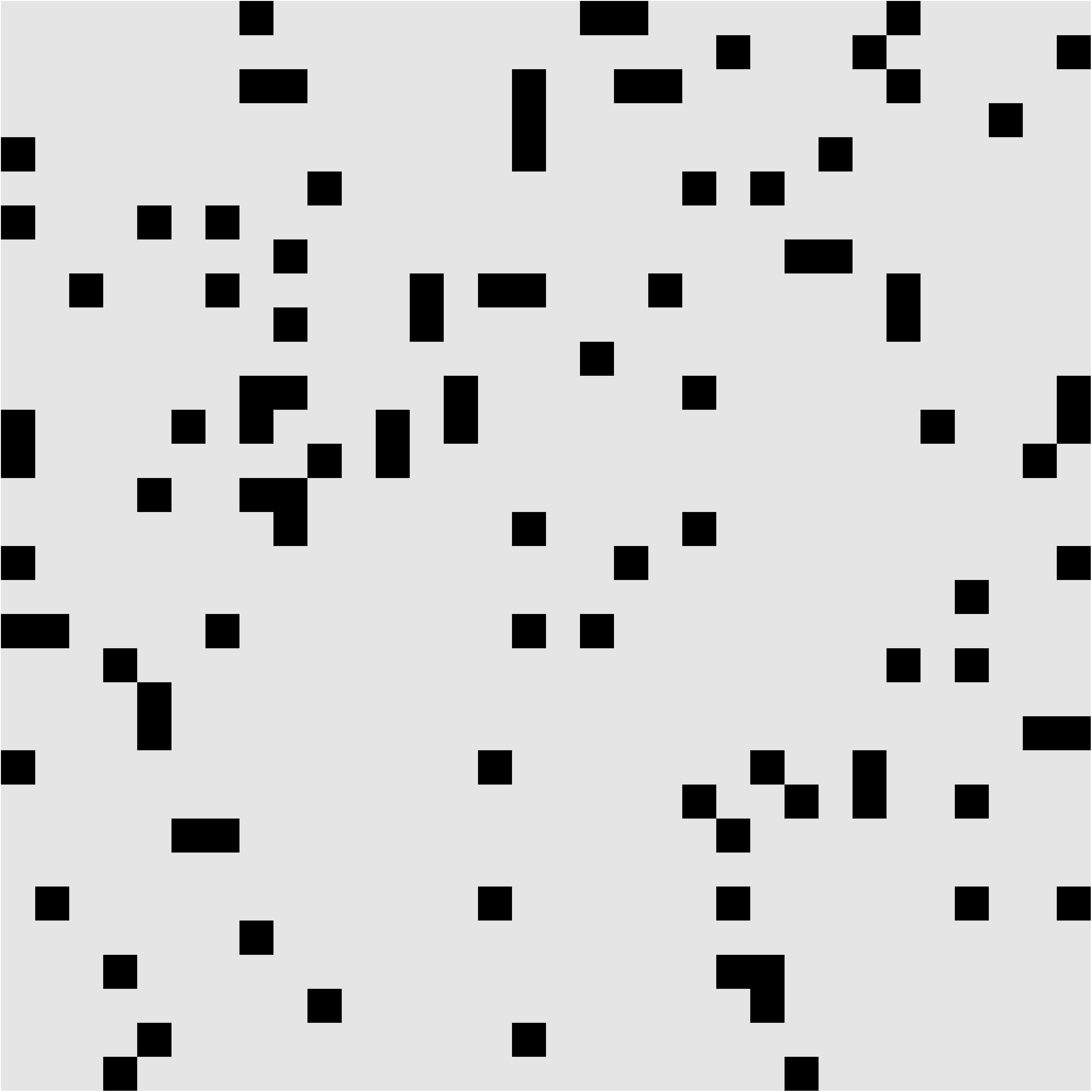}};
        \node[anchor=west,inner sep=0] at (4.8, 3.45) {\scriptsize random-32-32-10};
        \node[anchor=west,inner sep=0] at (4.8, 3.15) {\scriptsize $32\stimes 32$};
        \node[anchor=west,inner sep=0] at (4.8, 2.85) {\scriptsize $35$ agents};
      \end{tikzpicture}
    \end{minipage}
    \\
    \begin{minipage}{\colwidth}
      \centering\medskip
      \begin{tikzpicture}
        \node[anchor=south west,inner sep=0] at (0,0)
        {\includegraphics[width=1\hsize]{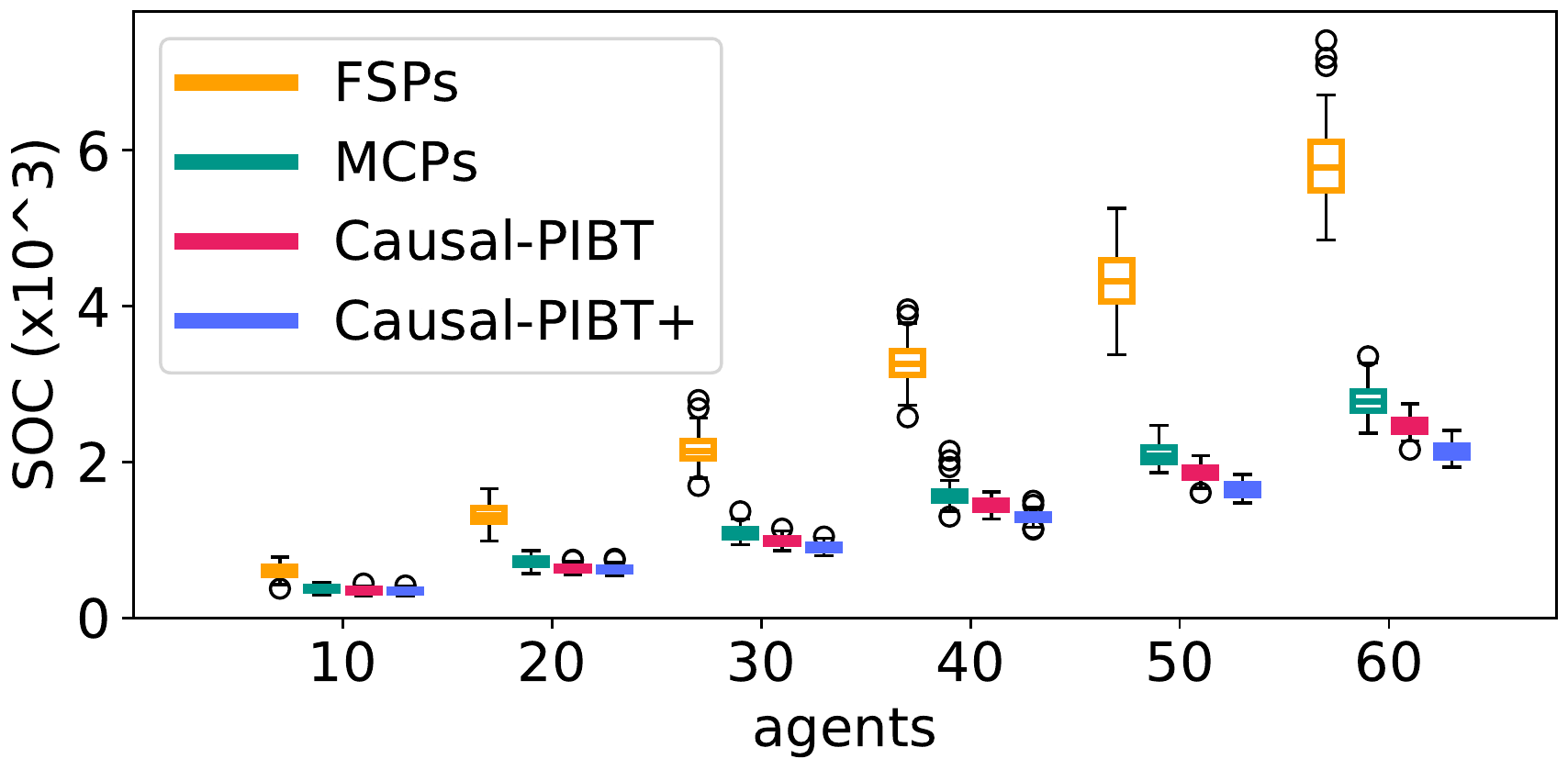}};
        \node[anchor=south west,inner sep=0] at (3.55,2.6)
        {\includegraphics[width=0.13\hsize]{fig/random-32-32-10.pdf}};
        \node[anchor=west,inner sep=0] at (4.65, 3.55) {\scriptsize random-32-32-10};
        \node[anchor=west,inner sep=0] at (4.65, 3.25) {\scriptsize $32\stimes 32$};
        \node[anchor=west,inner sep=0] at (4.65, 2.95) {\scriptsize $\bar{p}=0.5$};
      \end{tikzpicture}
    \end{minipage}
  \end{tabular}
  \caption{Executions in random grids.}
  \label{fig:result:2}
\end{figure}
}
{
  \newcommand{\colwidth}{0.9\hsize}
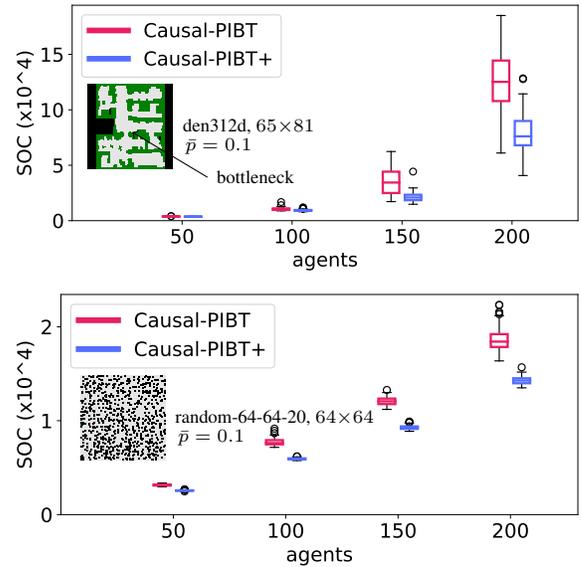
\begin{figure}[t]
  \centering
  \begin{tabular}{@{}c@{}}
    \begin{minipage}{\colwidth}
      \centering
      \begin{tikzpicture}
        \node[anchor=south west,inner sep=0] at (0,0)
        {\includegraphics[width=1\hsize]{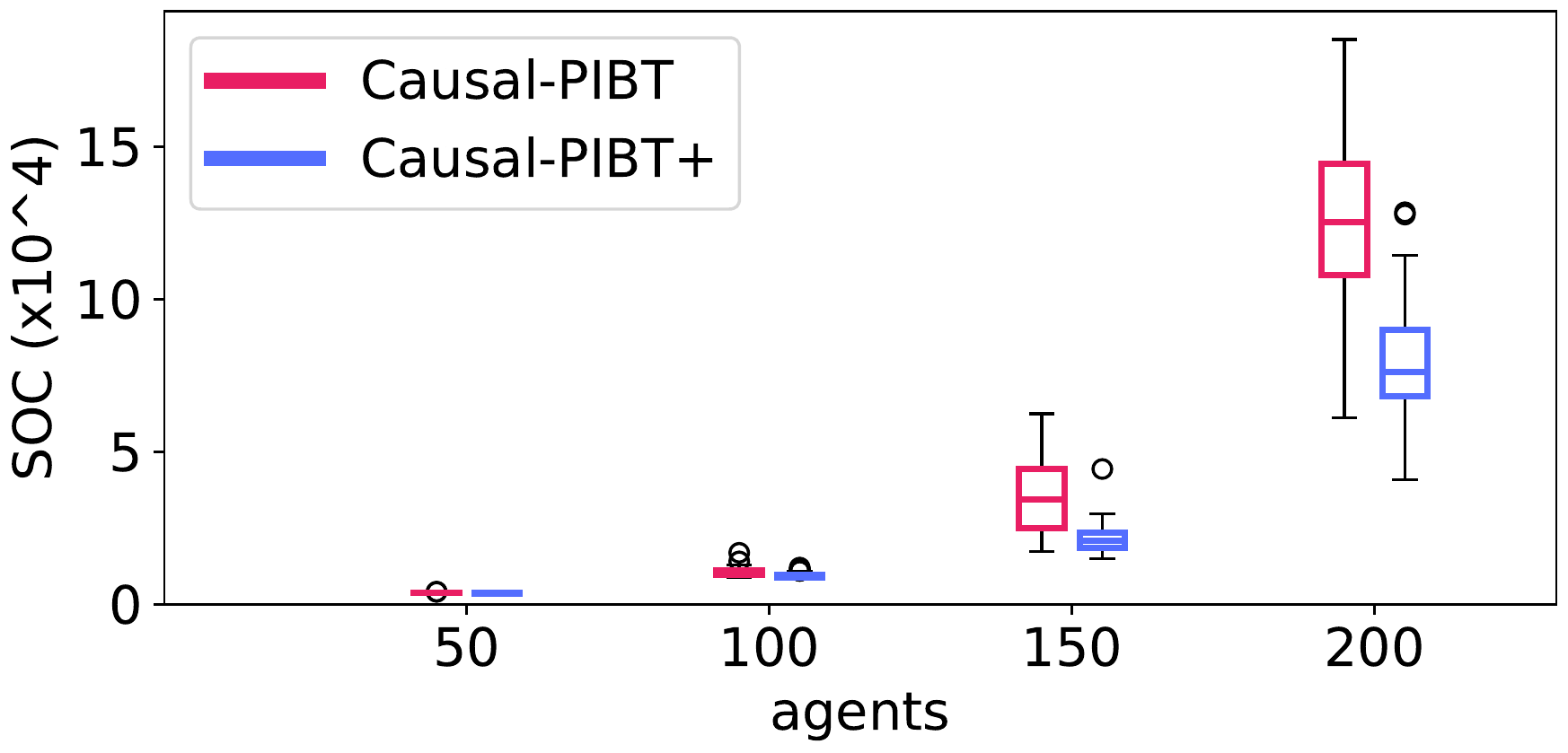}};
        \node[anchor=south west,inner sep=0,label=right:{\scriptsize den312d, $65\stimes 81$}] at (1.0,1.4)
        {\includegraphics[width=0.15\hsize]{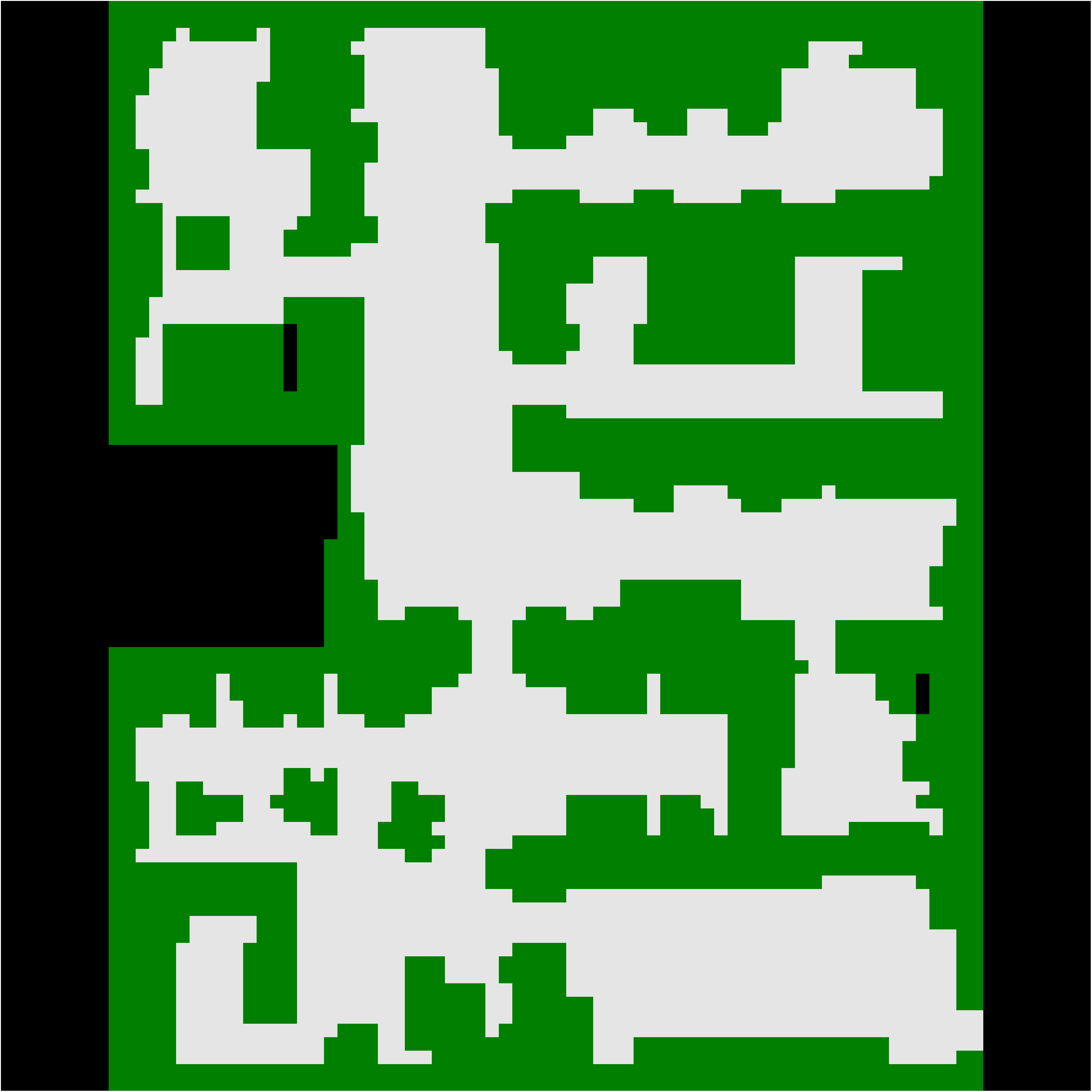}};
        \node[anchor=west](text1) at (2.6, 1.3) {\scriptsize bottleneck};
        \draw[->] (text1.west) -- (1.6, 1.9);
        \node[anchor=south west,inner sep=0] at (2.3,1.6) {\scriptsize $\bar{p}=0.1$};
      \end{tikzpicture}
    \end{minipage}
    \\
    \begin{minipage}{\colwidth}
      \centering\medskip
      \begin{tikzpicture}
        \node[anchor=south west,inner sep=0] at (0,0)
        {\includegraphics[width=1\hsize]{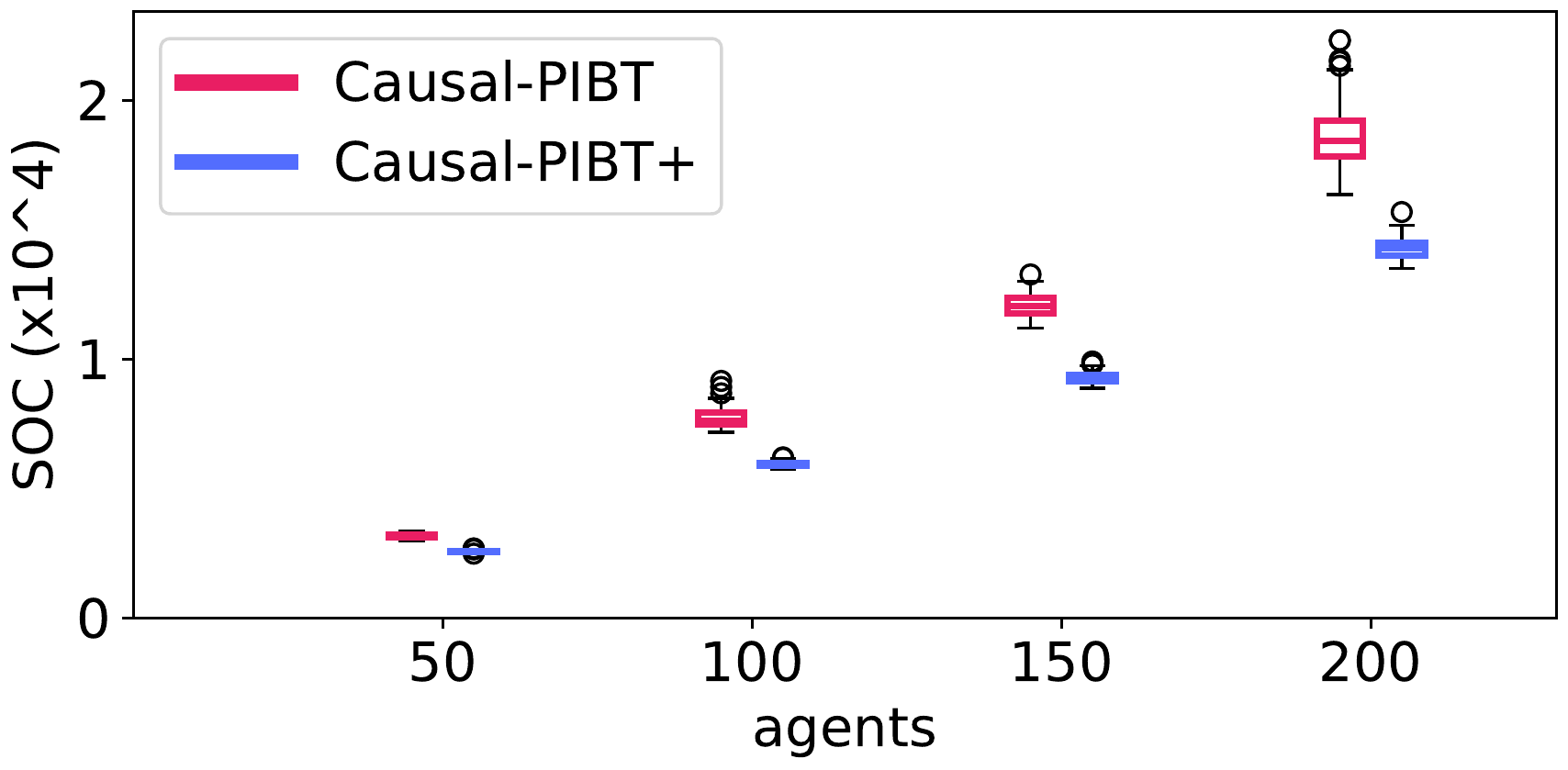}};
        \node[anchor=south west,inner sep=0,label=right:{\scriptsize random-64-64-20, $64\stimes 64$}] at (0.9,1.45)
        {\includegraphics[width=0.15\hsize]{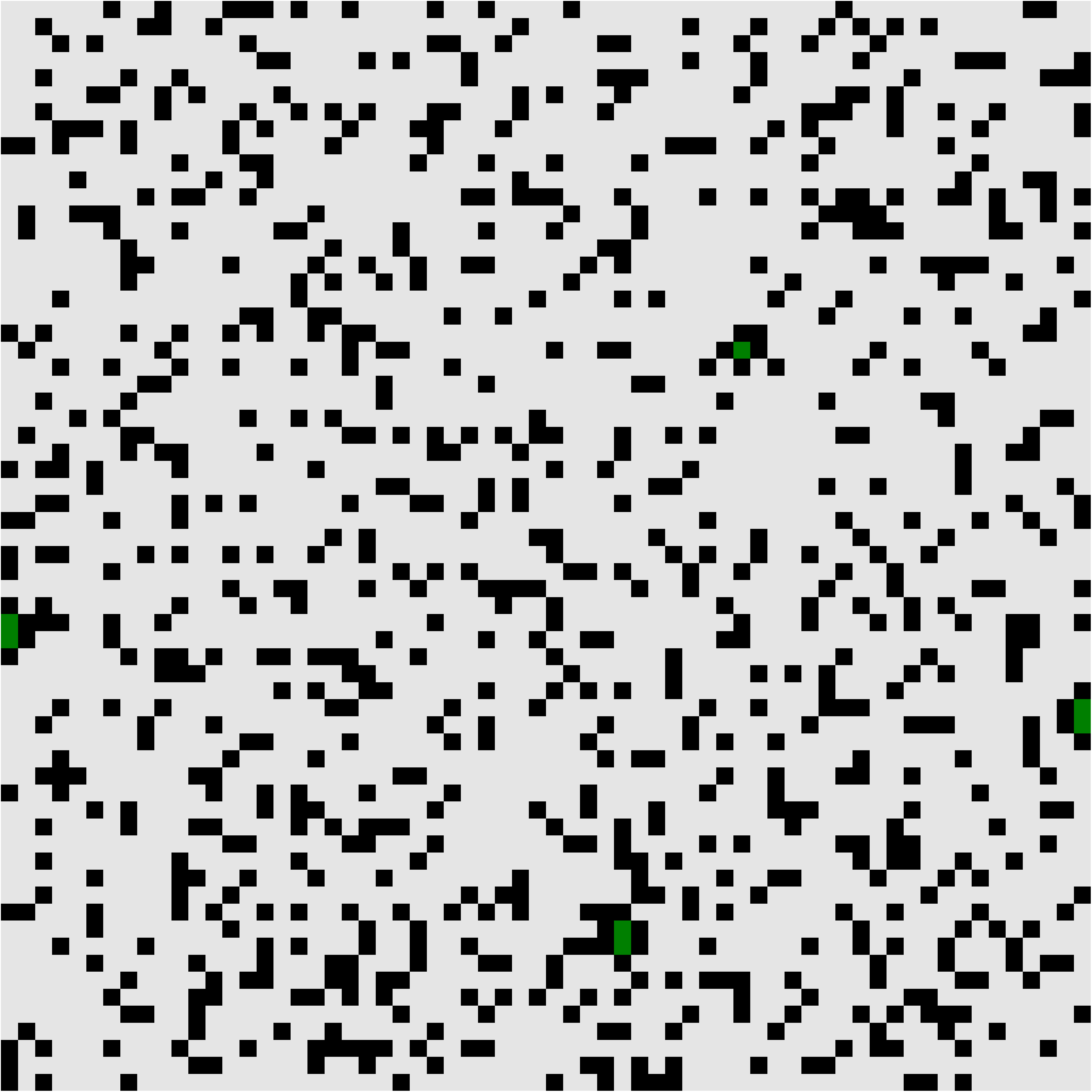}};
        \node[anchor=south west,inner sep=0] at (2.2,1.65) {\scriptsize $\bar{p}=0.1$};
      \end{tikzpicture}
    \end{minipage}
  \end{tabular}
  \caption{Executions in large fields.}
  \label{fig:result:3}
\end{figure}
}

\section{Evaluation}
\label{sec:evaluation}
The experiments aim at:
1)~showing the robustness of time-independent planning, i.e.,
even though delays of agents are unpredictable, their trajectories are kept efficient,
and,
2)~verifying the usefulness to use MAPF plans as hints during executions.

We used the MAPF-DP problem because it can emulate imperfect executions.
To adapt the time-independent model to MAPF-DP, rules of activation are required.
We repeated the following two phases:
1) Each agent $a_i$  in \extended is activated with probability $1-p_i$.
As a result, $a_i$ successfully moves to \head{i} with probability $1-p_i$ and becomes \contracted.
2) Pickup one agent in \contracted or in \requesting randomly then makes it activated, and repeat this until the configuration becomes stable, i.e, all agents in \contracted and \requesting do not change their states unless any agent in \extended is activated.
The rationale is that time required by executing atomic actions except for agents' moves is much smaller than that of the moves.
We regard a pair of these two phases as one timestep.
The strong termination was used to compare with existing two approaches for MAPF-DP: FSPs (Fully Synchronized Policies) and MCPs (Minimum Communication Policies).
The evaluation is based on the sum of cost (SOC) metric.
The simulator was developed in C++
\footnote{
  \url{https://github.com/Kei18/time-independent-planning}
},
and all experiments were run on a laptop with Intel Core i9 2.3GHz CPU and 16GB RAM.
In all settings, we tried $100$ repetitions.

\paragraph{Small Benchmarks}
First, we tested the time-independent planning in two small benchmarks shown in Fig.~\ref{fig:result:1}.
The delay probabilities $p_i$ were chosen uniformly at random from $[0, \bar{p}]$, where $\bar{p}$ is the upper bound of $p_i$.
The higher $\bar{p}$ means that agents delay often.
$\bar{p}\sequal 0$ corresponds to perfect executions without delays.
We manipulated $\bar{p}$.
\greedy, \algoname, and the enhanced one using MAPF plans as hints (\enhanced) were performed.
Those three were regarded to fail after $10000$ activations in total, implying occurring deadlocks or livelocks.
FSPs and MCPs were also tested as comparisons.
To execute FSPs, MCPs, and \enhanced, valid MAPF plans such that minimize SOC, which assumes following conflicts, were computed by an adapted version of Conflict-based Search (CBS)~\cite{sharon2015conflict}, prior to performing MAPF-DP.

The results are shown in Fig.~\ref{fig:result:1}.
Although \greedy failed in most cases due to deadlocks, \algoname(+) succeeded in all trials due to the deadlock-recovery and the reachability.
FSPs resulted in larger SOC compared to MCPs.
The results of \algoname(+) were competitive to those of MCPs.

\paragraph{Random Grid}
Next, we tested \algoname(+) using one scenario from MAPF benchmarks~\cite{stern2019def}, shown in Fig.~\ref{fig:result:2}.
We manipulated two factors:
1) changing $\bar{p}$ while fixing the number of agents ($\sequal 35$), and,
2) changing the number of agents while fixing $\bar{p}$ ($\sequal 0.5$).
When the number of agents increases, the probability that someone delays also increases.
We set sufficient upper bounds of activations.
FSPs and MCPs were also tested.
In this time, an adapted version of Enhanced CBS (ECBS)~\cite{barer2014suboptimal}, bounded sub-optimal MAPF solver, was used to obtain valid MAPF-DP plans, where the suboptimality was $1.1$.

Fig.~\ref{fig:result:2} shows the results.
The proposals succeeded in all cases even though the fields were not biconnected.
The results demonstrated the time-independent planning outputs robust executions while maintaining the small SOC in the severe environment as for timing assumptions.
In precisely, when $\bar{p}$ or the number of agents is small, MCPs was efficient, on the other hand, when these number increases, the time-independent plannings had an advantage.

\paragraph{Large Fields}
Finally, we tested proposals using large fields from the MAPF benchmarks, as shown in Fig.~\ref{fig:result:3}.
We respectively picked one scenario for each field, then tested while changing the number of agents.
$\bar{p}$ is fixed to $0.1$.
Regular ECBS (suboptimality: $1.1$) was used for \enhanced.
Note that \enhanced does not require MAPF plans that assume following conflicts.

The results (in Fig.~\ref{fig:result:3}) demonstrate the advantage of using MAPF plans in \algoname.
The proposed methods succeeded in all cases.
We observed a very large SOC in map den312d due to the existence of a bottleneck.
Such a structure critically affects execution delays.

\section{Conclusion}
\label{sec:conclusion}
This paper studied the online and distributed planning for multiple moving agents without timing assumptions.
We abstracted the reality of the execution as a transition system, then proposed time-independent planning, and \algoname as an implementation that ensures reachability.
Simulations in MAPF-DP demonstrate the robustness of time-independent planning and the usefulness of using MAPF plans as hints.

Future directions include the following:
1)~Develop time-independent planning with strong termination, e.g., by adapting Push and Swap~\cite{luna2011push} to our model.
2)~Address communication between agents explicitly.
In particular, this paper neglects delays caused by communication by treating interactions as a black box, and the next big challenge is there.

\section*{Acknowledgments}
We thank the anonymous reviewers for their many insightful comments and suggestions.
This work was partly supported by JSPS KAKENHI Grant Numbers~20J23011 and~20K11685, and Yoshida Scholarship Foundation.

{
\fontsize{9.0pt}{10.0pt} \selectfont 
\bibliography{ref}
}
\section{Appendix: Proof of Reachability}
In this section, we formally derive Theorem~\ref{theorem:reachability}.
We will use $V(G^\prime)$ and $E(G^\prime)$ as vertices and edges on a graph $G^\prime$ to clarify the graph that we mention.

At first, let us define a subgraph $H$ of $G$, called a \emph{parent-child graph}, to grasp the relationship between agents in \algoname.
Although $H$ is defined by a configuration formally, we just denote $H$ when we mention an invariant property of the parent-child graph.
\begin{definition}[parent-child graph]
Given a configuration $\gamma^k$, a \emph{parent-child graph} $H^k$ is a subgraph of $G$ with the following properties.
Let $B^k \subseteq A$ be a set of agents occupying only one node in $\gamma^k$, i.e., $B^k = \{ a_j \;|\; a_j \in A, \mode{j} \neq \extended \}$.
Then, $H^k = (V(H^k), E(H^k))$, where $V(H^k) = \{ \tail{j} \;|\; a_j \in B^k \}$ and $E(H^k) = \{ (\tail{j}, \tail{k}) \;|\; a_j \in B^k, a_k \in \children{j} \}$.
\end{definition}

The next lemma states that $H$ transits while keeping its graph property; forest.
\begin{lemma}
 \label{lemma:forest}
 $H$ is always a forest.
\end{lemma}
\begin{proof}
  Given the initial configuration $\gamma^0$, $H^0$ is obviously a forest.
  Assume that $H^{k-1}$ is a forest, where $k > 0$.
  Assume further that $a_i$ is activated and the configuration transits from $\gamma^{k-1}$ to $\gamma^k$.
  We now show that $H^k$ is also a forest.
  There are three kinds of causes for changing edges in $H^k$ from $H^{k-1}$.

  The first case is the procedure \textsc{ReleaseChildren}.
  This removes all edges in $H^{k-1}$ from $\tail{i}$ to its children's tail.
  This is a decomposition of a subtree in $H^{k-1}$, thus, $H^k$ is still a forest.

  The second case is the procedure \textsc{PriorityInheritance}.
  If there exists an agent $a_k$ as a new parent, this procedure first $\tail{i}$ be isolated in $H^{k-1}$.
  Then, $a_i$ updates its parent, i.e., a new edge between \tail{i} and \tail{k} is added to $H^k$.
  This update still keeps $H^k$ being a forest.
  Note that cycles are not constructed;
  agents in one component have the same \ptmp{} due to priority inheritance.
  Priority inheritance occurs only between agents with different \ptmp{}.

  The third case is that $a_i$ transits from \requesting to \extended.
  Before being \extended, $a_i$ makes \tail{i} isolated.
  This is a decomposition of a subtree in $H^{k-1}$, thus, $H^k$ is kept as a forest.

  By induction, $H$ is always a forest.
\end{proof}

We next introduce a \emph{dominant agent/tree}.
Intuitively, it is an agent or a tree with the highest priority in one configuration.
Note that there exists a configuration without a dominant agent/tree, e.g., when the agent with highest \pori{} is \extended.
\begin{definition}[dominant agent/tree]
 Assume an agent $a_d$ in \\ \contracted such that $\pori{d} > \pori{j}, \forall a_j \in A \setminus \{ a_d \}$.
 $a_d$ is a \emph{dominant agent} when it becomes \requesting, until it becomes \extended.
 A component of the parent-child graph $H$ containing \tail{d} is a \emph{dominant tree} $T_d$.
\end{definition}

Since $H$ is a forest from Lemma~\ref{lemma:forest}, $T_d$ is a tree.
Subsequently, we state the condition where the dominant agent can move to its desired node, then derive reachability.
\begin{definition}[movable agent]
 An agent $a_m$ is said to be \emph{movable} in a given configuration when $\mode{m} = \requesting$ and $\lnot\occupied{\head{m}}$.
\end{definition}

\begin{lemma}
 \label{lemma:movable}
 If there is a dominant tree $T_d$ and a movable agent $a_m$ such that $\tail{m} \in V(T_d)$, then a dominant agent $a_d$ eventually becomes \extended without being \contracted.
\end{lemma}
\begin{proof}
 Only the procedure \textsc{PriorityInheritance} adds new edges to $H$ as shown in the proof of Lemma~\ref{lemma:forest}.
 Thus, all agents in $T_d$ have the same temporal priority, i.e., \ptmp{d}.
 This means that no agents outside of $T_d$ can disrupt $T_d$.

 Considering the construction of $T_d$, all agents on a path from $\tail{d}$ to $\tail{m}$ in $T_d$ are \requesting.
 The other agents in $T_d$ are \contracted and do nothing if they are activated.
 Moreover, if agents in \requesting in $T_d$ excluding $a_m$ are activated, they do nothing.
 As a result, only the activation of $a_m$ can update states of agents in $T_d$.

 Let $a_{m^\prime}$ be $\parent{m}$.
 Any agent $a_i$ in \requesting with $\head{i} = \head{m}$ will defeat since $\ptmp{m} = \ptmp{d}$.
 As a result, when $a_m$ is activated, it certainly becomes \extended and departs from $T_d$.
 $a_m$ becomes \contracted in finite activations.
 During this transition, any agent $a_i$ in \requesting with $\head{i} = \head{m^\prime}$ will defeat since $\ptmp{m^\prime} = \ptmp{d}$.
 Now, $\head{m^\prime}$ is not occupied by any agents.
 We have a new movable agent, namely, $a_{m^\prime}$.
 By the same procedure, $a_{m^\prime}$ moves to its head in finite time.
 By induction, $a_d$ eventually becomes \extended without being $\contracted$.
\end{proof}

\begin{lemma}
 \label{lemma:dominant}
 If $G$ is biconnected and $|A| < |V|$, a dominant agent $a_d$ eventually becomes \extended without being \contracted.
\end{lemma}
\begin{proof}
 First, we show that $a_d$ does not become \contracted.
 Assume by contradiction that $a_d$ becomes \contracted.
 Considering the process of the construction of the dominant tree $T_d$, all agents $a_i$ in $T_d$ excluding $a_d$ are \contracted and $C_i = \emptyset$.
 This implies that there exists no neighbor nodes of $V(T_d) \setminus \{ \tail{d} \}$ in $V(G)$.
 Since $G$ is biconnected, although \tail{d} is removed from $G$, there exists a path between any two nodes in $V(G)$.
 Thus, $V(T_d)$ must spans $G$, i.e., $|V| = |A|$.
 This is a contradiction, so, $a_d$ never becomes \contracted.

 Since $T_d$ tries to expand so that it spans $G$, and $|A| < |V|$, we have eventually a movable agent in $T_d$.
 According to lemma~\ref{lemma:movable}, $a_d$ eventually becomes \extended.
\end{proof}

\setcounter{theorem}{1}
\begin{theorem}[reachability]
 \algoname ensures reachability in biconnected graphs if $|A| < |V|$.
\end{theorem}
\begin{proof}
 Let an agent be $a_h$ with the highest \pori{}.
 If $a_h$ is \contracted, this agent does not inherit any priority.
 Thus, $a_h$ can be \requesting so that \head{h} is any neighbor nodes of \tail{h}.
 If $a_h$ is \requesting, $a_h$ is a dominant agent.
 According to lemma~\ref{lemma:dominant}, $a_h$ eventually moves to \head{h}.
 These imply that $a_h$ can move to an arbitrary neighbor node in finite activations.
 By this, $a_h$ moves along the shortest path to its goal, then drops its priority $\pori{h}$.

 Due to the above mechanism, an agent that has not reached its goal eventually gets the highest \pori{}, then it starts moving along its shortest path to its goal.
 This satisfies reachability.
\end{proof}

It is interesting to note that the necessary condition to ensure reachability in \algoname is more restrictive than in PIBT. This is due to the prohibition of rotations (cyclic conflicts) by the time-independent model.

\end{document}